%% file: main.tex
\documentclass[manuscript,authorversion,nonacm]{acmart}
\pdfoutput=1

\usepackage[utf8]{inputenc} 
\usepackage[T1]{fontenc}    
\usepackage{hyperref}       
\usepackage{url}            
\usepackage{booktabs}       
\usepackage{amsfonts}       
\usepackage{nicefrac}       
\usepackage{microtype}

\usepackage{amsmath}
\usepackage{amsthm}
\usepackage{natbib}
\usepackage{mathtools}
\usepackage{graphicx}
\usepackage{subcaption}
\usepackage{mwe}
\usepackage{algorithm}
\usepackage[noend]{algpseudocode}
\usepackage{nicefrac}
\usepackage{multirow}
\usepackage{multicol}
\usepackage{paralist}
\DeclarePairedDelimiter{\ceil}{\lceil}{\rceil}

\usepackage[inline]{enumitem}

\AtBeginDocument{%
  \providecommand\BibTeX{{%
    \normalfont B\kern-0.5em{\scshape i\kern-0.25em b}\kern-0.8em\TeX}}}

\newcommand{\etal}{\textit{et~al. }}

\newcommand{\cost}{\textrm{cost}}

\DeclareMathOperator*{\argmax}{arg\,max}
\DeclareMathOperator*{\argmin}{arg\,min}

\newtheorem{thm}{Theorem}
\theoremstyle{definition}
\newtheorem{defn}[thm]{Definition}

\usepackage[draft,inline,nomargin,index]{fixme}

\fxsetup{theme=color,mode=multiuser,inlineface=\itshape,envface=\itshape}
\FXRegisterAuthor{ma}{ama}{\colorbox{blue!10!white}{\color{black}Mohsen}}
\FXRegisterAuthor{sv}{asv}{\colorbox{gray!10!white}{\color{black}Suresh}}
\FXRegisterAuthor{sf}{asf}{\colorbox{red!10!white}{\color{black}Sorelle}}
\FXRegisterAuthor{kl}{akl}{\colorbox{purple!10!white}{\color{black}Kristian}}
\fxusetargetlayout{color}

\begin{document}

\title[Measuring and mitigating voting access disparities]{Measuring and mitigating voting access disparities: a study of race and polling locations in Florida and North Carolina}


\author{Mohsen Abbasi}
\email{mohsena@cs.utah.edu}
\affiliation{
  \institution{University of Utah}
  \country{USA}
}

\author{Suresh Venkatasubramanian}
\email{suresh@brown.edu}
\affiliation{%
  \institution{Brown University}
  \country{USA}
}

\author{Sorelle A. Friedler}
\email{sorelle@cs.haverford.edu}
\affiliation{
  \institution{Haverford College}
  \country{USA}
}

\author{Kristian Lum}
\email{kristianlum@gmail.com}
\affiliation{%
  \institution{University of Pennsylvania}
  \country{USA}
}

\author{Calvin Barrett}
\email{cpbarrett@haverford.edu}
\affiliation{%
  \institution{Haverford College}
  \country{USA}
}

\renewcommand{\shortauthors}{Abbasi, et al.}

\begin{abstract}

Voter suppression and associated racial disparities in access to voting are
long-standing civil rights concerns in the United States. Barriers to voting
have taken many forms over the decades.  A history of violent explicit
discouragement has shifted to more subtle access limitations that can include
long lines and wait times, long travel times to reach a polling station, and
other logistical barriers to voting. Our focus in this work is on quantifying disparities in
voting access pertaining to the overall time-to-vote, and how they could be remedied
via a better choice of polling location or provisioning more sites where voters
can cast ballots. However, appropriately calibrating access disparities is
difficult because of the need to account for factors such as population density
and different community expectations for reasonable travel times. 

In this paper, we quantify access to polling locations, developing a methodology
for the calibrated measurement of racial disparities in polling location "load"
and distance to polling locations.  We apply this methodology to a study of
real-world data from Florida and North Carolina to identify disparities in
voting access from the 2020 election.  We also introduce algorithms, with
modifications to handle scale, that can reduce these disparities by suggesting
new polling locations from a given list of identified public locations
(including schools and libraries).  Applying these algorithms on the 2020
election location data also helps to expose and explore tradeoffs between the
cost of allocating more polling locations and the potential impact on access
disparities.  The developed voting access measurement methodology and
algorithmic remediation technique is a first step in better polling location
assignment.


\end{abstract}





\maketitle

\input{intro}
\input{data}
\input{stats}
\input{solution}
\input{experiments}
\input{conclusion}



\newpage
\bibliographystyle{ACM-Reference-Format}
\bibliography{refs}

\appendix

\section{Coreset analysis}
A comparison between the results of running regular and fair $k$-median algorithms on a sample dataset and its coreset are presented in Figure~\ref{fig:coreset}. We observe the objective values achieved by running both algorithms on the coreset are very close to the values from the entire sampled data, which corroborates the results of Theorem~\ref{th:coreset}.

\begin{figure*}[htb]
\centering
\includegraphics[width=0.8\textwidth]{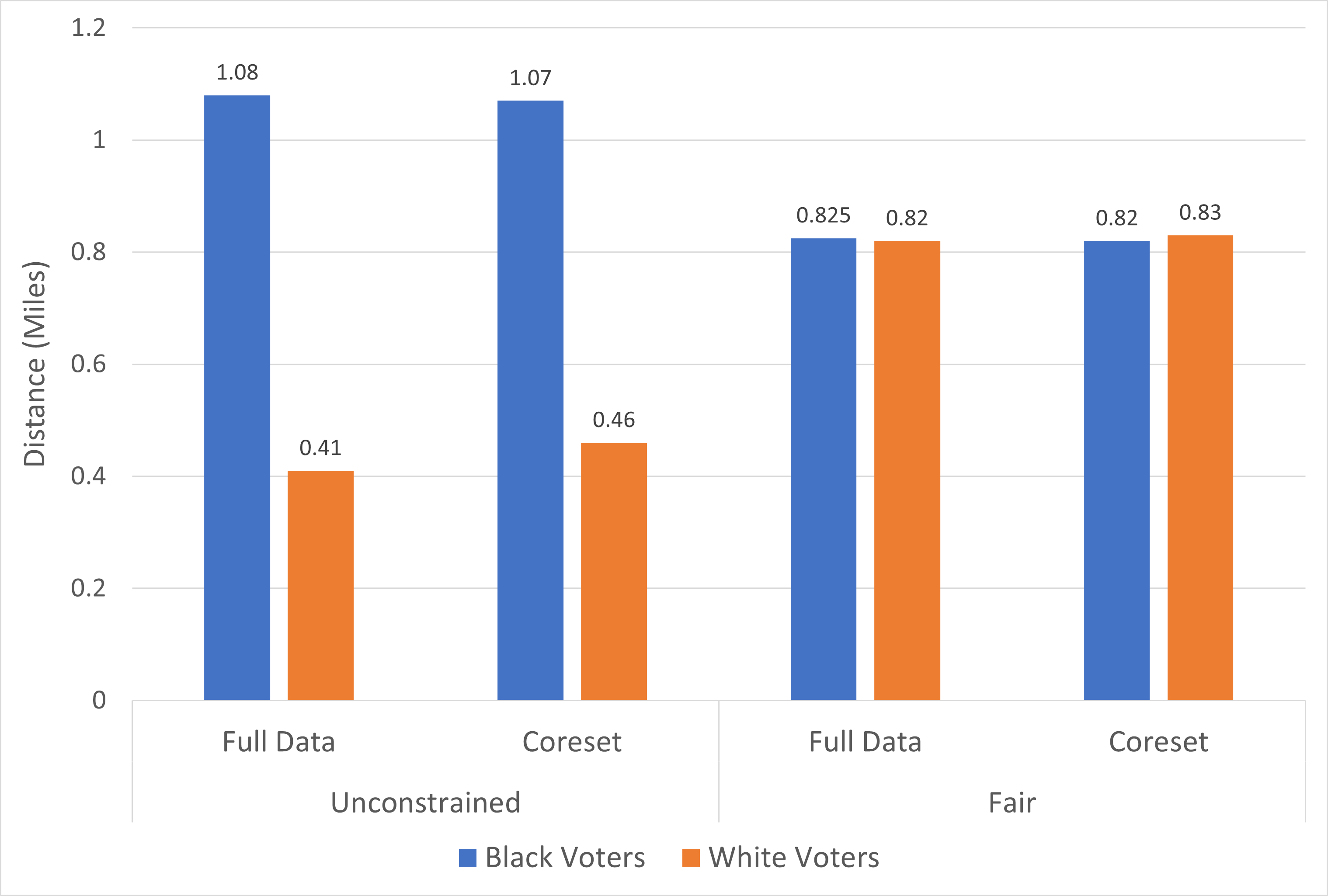}
\caption
{Average distance to polling locations for White and Black voters in sample data from North Carolina records. Results for both regular and fair $k$-median algorithms on the constructed coreset (with $\epsilon=0.1$). We observe are competitive with their corresponding values for the entire sample.}
\label{fig:coreset}
\end{figure*}

\section{Distances to polling place in alternative voting location assignments}
In this section we present the results of the proposed algorithms with respect to the distances between voters of each racial group and their assigned polling locations.

\subsection{Fair $k$-median}
\subsubsection{Absolute distances}
Tables~\ref{tab:fl_fair_dist} and \ref{tab:nc_fair_dist} summarize the absolute distances for the assignment by the fair $k$-median algorithm for Florida and North Carolina, respectively.

\begin{table}[H]
\caption{Florida: distance to polling locations in Fair $k$-median  assignment}
\label{tab:fl_fair_dist}
\centering
\begin{tabular}{@{}lcccccc@{}}
\toprule
\multicolumn{1}{c}{\multirow{3}{*}{Race}} & \multicolumn{3}{c}{Schools and Libraries} & \multicolumn{3}{c}{\begin{tabular}[c]{@{}c@{}}All data points\end{tabular}} \\ \cmidrule(l){2-4} \cmidrule(l){5-7}
\multicolumn{1}{c}{}    & Mean      & Median      & Max        & Mean   & Median    & Max   \\
\midrule
White, not Hispanic         & 0.91   & 0.68   & 24.27    & 0.77   & 0.58     & 23.25  \\
Hispanic                    & 0.58   & 0.46   & 16.45    & 0.49   & 0.39     & 15.74   \\
Black, not Hispanic         & 0.60   & 0.48   & 24.27    & 0.52   & 0.41     & 23.25  \\
Unknown                     & 0.66   & 0.51   & 24.27    & 0.55   & 0.43     & 23.25    \\
Asian or Pacific Islander   & 0.70   & 0.59   & 16.07    & 0.59   & 0.49     & 13.53     \\
Other                       & 0.72   & 0.56   & 15.69    & 0.60   & 0.47     & 11.50     \\
Multi-Racial                & 0.69   & 0.55   & 14.31    & 0.57   & 0.45     & 10.78    \\
American Indian             & 0.88   & 0.62   & 17.06    & 0.67   & 0.49     & 12.61    \\ \bottomrule
\end{tabular}
\end{table}

\begin{table}[H]
\caption{North Carolina: distance to polling locations in Fair $k$-median  assignment}
\label{tab:nc_fair_dist}
\centering
\begin{tabular}{@{}lcccccc@{}}
\toprule
\multicolumn{1}{c}{\multirow{3}{*}{Race}} & \multicolumn{3}{c}{Schools and Libraries (Miles)} & \multicolumn{3}{c}{\begin{tabular}[c]{@{}c@{}}All data points (Miles)\end{tabular}} \\ \cmidrule(l){2-4} \cmidrule(l){5-7}
\multicolumn{1}{c}{}    & Mean      & Median      & Max        & Mean   & Median    & Max   \\
\midrule
White    & 1.33      & 1.00    & 13.66      & 1.14      & 0.82  & 13.66  \\
Black or African American   & 1.06      & 0.77        & 11.40    & 0.90   & 0.65    & 11.46   \\
Undesignated    & 1.17      & 0.86        & 11.48      &    0.97   &   0.69  &   11.47  \\
Other   & 1.06      & 0.80        & 10.79      & 0.85   & 0.64      & 10.87    \\
Asian   & 0.88      & 0.72        & 9.45      & 0.68   & 0.55      & 9.39     \\
American Indian or Alaska Native    & 1.65      & 1.28        & 9.61      & 1.31   & 1.01      & 10.02     \\
Two or more races   & 1.01      & 0.76        & 9.52      & 0.80   & 0.60  & 9.69    \\
Native Hawaiian  or Pacific Islander    & 1.18      & 0.79        & 6.43       & 0.82   & 0.58       & 4.26    \\ \bottomrule
\end{tabular}
\end{table}

\subsubsection{Normalized distances}
Tables~\ref{tab:fl_fair_dist_normalized} and \ref{tab:nc_fair_dist_normalized} summarize the normalized median distances for the assignment by the fair $k$-median algorithm. Columns ``Schools and Libraries'' and ``All data points'' indicate the valid polling locations considered in the solution, and sub-columns ``Nearest school/library'' and ``Median pairwise voter distance'' indicate the normalization baseline.

\begin{table}[H]
\caption{Florida: Median normalized distances to polling locations in Fair $k$-median  assignment}
\label{tab:fl_fair_dist_normalized}
\centering
\resizebox{\textwidth}{!}{
\begin{tabular}{@{}lcccccc@{}}
\toprule
\multicolumn{1}{c}{\multirow{3}{*}{Race}} & \multicolumn{2}{c}{Schools and Libraries} & \multicolumn{2}{c}{\begin{tabular}[c]{@{}c@{}}All data points\end{tabular}} \\ \cmidrule(l){2-3} \cmidrule(l){4-5}
\multicolumn{1}{c}{}    & Nearest school/library      & Median pairwise voter distance        & Nearest school/library      & Median pairwise voter distance   \\
\midrule
White, not Hispanic         & 1   & 0.77   & 0.71    & 0.67 \\
Hispanic                    & 1   & 0.78   & 0.76    & 0.67 \\
Black, not Hispanic         & 1   & 0.76   & 0.82    & 0.66 \\
Unknown                     & 1   & 0.76   & 0.75    & 0.65 \\
Asian or Pacific Islander   & 1   & 0.75   & 0.74    & 0.64 \\
Other                       & 1   & 0.76   & 0.73    & 0.65 \\
Multi-Racial                & 1   & 0.75   & 0.74    & 0.64 \\
American Indian             & 1   & 0.74   & 0.69    & 0.60  \\ \bottomrule
\end{tabular}
}
\end{table}

\begin{table}[H]
\caption{North Carolina: Median normalized distances to polling locations in Fair $k$-median  assignment}
\label{tab:nc_fair_dist_normalized}
\centering
\resizebox{\textwidth}{!}{
\begin{tabular}{@{}lcccccc@{}}
\toprule
\multicolumn{1}{c}{\multirow{3}{*}{Race}} & \multicolumn{2}{c}{Schools and Libraries} & \multicolumn{2}{c}{\begin{tabular}[c]{@{}c@{}}All data points\end{tabular}} \\ \cmidrule(l){2-3} \cmidrule(l){4-5}
\multicolumn{1}{c}{}    & Nearest school/library      & Median pairwise voter distance        & Nearest school/library      & Median pairwise voter distance   \\
\midrule
White                               & 1   & 0.79   & 0.78    & 0.68 \\
Black or African American           & 1   & 0.82   & 0.87    & 0.69 \\
Undesignated                        & 1   & 0.80   & 0.78    & 0.66 \\
Other                               & 1   & 0.81   & 0.78    & 0.65 \\
Asian                               & 1   & 0.82   & 0.72    & 0.64 \\
American Indian or Alaska Native    & 1   & 0.75   & 0.75    & 0.59 \\
Two or more races                   & 1   & 0.81   & 0.76    & 0.65 \\
Native Hawaiian  or Pacific Islander& 1   & 0.74   & 0.60    & 0.50  \\ \bottomrule
\end{tabular}
}
\end{table}

\subsection{Balanced fair $k$-median}
\label{app:balanced}

\subsubsection{Absolute distances}
\label{sec:exp_balfair_med}

The results of $L$-balanced fair $k$-median algorithm for $\epsilon = 0.1$, $\epsilon = 0.5$ and $\epsilon = 0.9$ for the states of Florida and North Carolina are presented in Tables~\ref{tab:fl_balanced_fair_dist_1}, \ref{tab:nc_balanced_fair_dist_1}, \ref{tab:fl_balanced_fair_dist_5}, \ref{tab:nc_balanced_fair_dist_5}, \ref{tab:fl_balanced_fair_dist_9}, and \ref{tab:nc_balanced_fair_dist_9}.

\begin{table}[H]
\caption{Florida: distance to polling locations in Balanced Fair $k$-median  assignment for $\epsilon = 0.1$}
\label{tab:fl_balanced_fair_dist_1}
\centering
\begin{tabular}{@{}lcccccc@{}}
\toprule
\multicolumn{1}{c}{\multirow{3}{*}{Race}} & \multicolumn{3}{c}{Schools and Libraries} & \multicolumn{3}{c}{\begin{tabular}[c]{@{}c@{}}All data points\end{tabular}} \\ \cmidrule(l){2-4} \cmidrule(l){5-7}
\multicolumn{1}{c}{}    & Mean      & Median      & Max        & Mean   & Median    & Max   \\
\midrule
White, not Hispanic         & 1.12   & 0.73   & 24.27    & 0.77   & 0.58     & 23.25  \\
Hispanic                    & 0.69   & 0.48   & 16.45    & 0.48   & 0.39     & 15.74   \\
Black, not Hispanic         & 0.60   & 0.50   & 24.27    & 0.52   & 0.41     & 23.25  \\
Unknown                     & 0.79   & 0.54   & 24.27    & 0.55   & 0.43     & 23.25    \\
Asian or Pacific Islander   & 0.86   & 0.63   & 16.07    & 0.59   & 0.49     & 13.43     \\
Other                       & 0.88   & 0.60   & 15.69    & 0.60   & 0.47     & 11.50     \\
Multi-Racial                & 0.84   & 0.58   & 14.31    & 0.57   & 0.45     & 10.78    \\
American Indian             & 1.06   & 0.66   & 17.06    & 0.67   & 0.49     & 12.60    \\ \bottomrule
\end{tabular}
\end{table}

\begin{table}[H]
\caption{North Carolina: distance to polling locations in Balanced Fair $k$-median assignment for $\epsilon = 0.1$}
\label{tab:nc_balanced_fair_dist_1}
\centering
\begin{tabular}{@{}lcccccc@{}}
\toprule
\multicolumn{1}{c}{\multirow{3}{*}{Race}} & \multicolumn{3}{c}{Schools and Libraries (Miles)} & \multicolumn{3}{c}{\begin{tabular}[c]{@{}c@{}}All data points (Miles)\end{tabular}} \\ \cmidrule(l){2-4} \cmidrule(l){5-7}
\multicolumn{1}{c}{}    & Mean      & Median      & Max        & Mean   & Median    & Max   \\
\midrule
White                       & 1.59      & 1.11        & 25.99      & 1.14   & 0.82    & 13.66  \\
Black or African American   & 1.20      & 0.81        & 25.64      & 0.90   & 0.65    & 11.46   \\
Undesignated                & 1.41      & 0.94        & 25.99      & 0.97   & 0.70    & 11.47  \\
Other                       & 1.31      & 0.88        & 25.61      & 0.85   & 0.64    & 10.86    \\
Asian                       & 1.20      & 0.81        & 24.40       & 0.68   & 0.55    & 9.39     \\
American Indian or Alaska Native    & 1.73      & 1.35        & 24.78      & 1.31   & 1.01      & 10.02     \\
Two or more races           & 1.23      & 0.81        & 25.61      & 0.80   & 0.60  & 9.69    \\
Native Hawaiian  or Pacific Islander    & 1.62      & 0.95        & 13.45       & 0.82   & 0.58       & 4.26    \\ \bottomrule
\end{tabular}
\end{table}

\begin{table}[H]
\caption{Florida: distance to polling locations in Balanced Fair $k$-median  assignment for $\epsilon = 0.5$}
\label{tab:fl_balanced_fair_dist_5}
\centering
\begin{tabular}{@{}lcccccc@{}}
\toprule
\multicolumn{1}{c}{\multirow{3}{*}{Race}} & \multicolumn{3}{c}{Schools and Libraries} & \multicolumn{3}{c}{\begin{tabular}[c]{@{}c@{}}All data points\end{tabular}} \\ \cmidrule(l){2-4} \cmidrule(l){5-7}
\multicolumn{1}{c}{}    & Mean      & Median      & Max        & Mean   & Median    & Max   \\
\midrule
White, not Hispanic         & 0.97   & 0.69   & 24.27    & 0.77   & 0.58     & 23.25  \\
Hispanic                    & 0.61   & 0.46   & 16.45    & 0.48   & 0.39     & 15.74   \\
Black, not Hispanic         & 0.62   & 0.48   & 24.27    & 0.52   & 0.41     & 23.25  \\
Unknown                     & 0.69   & 0.51   & 24.27    & 0.55   & 0.43     & 23.25    \\
Asian or Pacific Islander   & 0.74   & 0.59   & 16.07    & 0.59   & 0.49     & 13.43     \\
Other                       & 0.75   & 0.56   & 15.69    & 0.60   & 0.47     & 11.50     \\
Multi-Racial                & 0.72   & 0.55   & 14.31    & 0.57   & 0.45     & 10.78    \\
American Indian             & 0.93   & 0.62   & 17.06    & 0.67   & 0.49     & 12.60    \\ \bottomrule
\end{tabular}
\end{table}

\begin{table}[H]
\caption{North Carolina: distance to polling locations in Balanced Fair $k$-median assignment for $\epsilon = 0.5$}
\label{tab:nc_balanced_fair_dist_5}
\centering
\begin{tabular}{@{}lcccccc@{}}
\toprule
\multicolumn{1}{c}{\multirow{3}{*}{Race}} & \multicolumn{3}{c}{Schools and Libraries (Miles)} & \multicolumn{3}{c}{\begin{tabular}[c]{@{}c@{}}All data points (Miles)\end{tabular}} \\ \cmidrule(l){2-4} \cmidrule(l){5-7}
\multicolumn{1}{c}{}    & Mean      & Median      & Max        & Mean   & Median    & Max   \\
\midrule
White                       & 1.40      & 1.03        & 18.99      & 1.14   & 0.82    & 13.66  \\
Black or African American   & 1.08      & 0.77        & 18.96      & 0.90   & 0.65    & 11.46   \\
Undesignated                & 1.23      & 0.88        & 18.97      & 0.97   & 0.70    & 11.47  \\
Other                       & 1.12      & 0.82        & 18.94      & 0.85   & 0.64    & 10.86    \\
Asian                       & 0.98      & 0.75        & 18.81       & 0.68   & 0.55    & 9.39     \\
American Indian or Alaska Native    & 1.67      & 1.31        & 18.96      & 1.31   & 1.01      & 10.02     \\
Two or more races           & 1.05      & 0.76        & 18.42      & 0.80   & 0.60  & 9.69    \\
Native Hawaiian  or Pacific Islander    & 1.29      & 0.78        & 11.75       & 0.82   & 0.58       & 4.26    \\ \bottomrule
\end{tabular}
\end{table}

\begin{table}[H]
\caption{Florida: distance to polling locations in Balanced Fair $k$-median  assignment for $\epsilon = 0.9$}
\label{tab:fl_balanced_fair_dist_9}
\centering
\begin{tabular}{@{}lcccccc@{}}
\toprule
\multicolumn{1}{c}{\multirow{3}{*}{Race}} & \multicolumn{3}{c}{Schools and Libraries} & \multicolumn{3}{c}{\begin{tabular}[c]{@{}c@{}}All data points\end{tabular}} \\ \cmidrule(l){2-4} \cmidrule(l){5-7}
\multicolumn{1}{c}{}    & Mean      & Median      & Max        & Mean   & Median    & Max   \\
\midrule
White, not Hispanic         & 0.92   & 0.67   & 24.27    & 0.77   & 0.58     & 23.25  \\
Hispanic                    & 0.59   & 0.45   & 16.45    & 0.48   & 0.39     & 15.74   \\
Black, not Hispanic         & 0.60   & 0.48   & 24.27    & 0.52   & 0.41     & 23.25  \\
Unknown                     & 0.66   & 0.51   & 24.27    & 0.55   & 0.43     & 23.25    \\
Asian or Pacific Islander   & 0.71   & 0.58   & 16.07    & 0.59   & 0.49     & 13.43     \\
Other                       & 0.72   & 0.56   & 15.69    & 0.60   & 0.47     & 11.50     \\
Multi-Racial                & 0.69   & 0.54   & 14.31    & 0.57   & 0.45     & 10.78    \\
American Indian             & 0.89   & 0.62   & 17.06    & 0.67   & 0.49     & 12.60    \\ \bottomrule
\end{tabular}
\end{table}

\begin{table}[H]
\caption{North Carolina: distance to polling locations in Balanced Fair $k$-median assignment for $\epsilon = 0.9$}
\label{tab:nc_balanced_fair_dist_9}
\centering
\begin{tabular}{@{}lcccccc@{}}
\toprule
\multicolumn{1}{c}{\multirow{3}{*}{Race}} & \multicolumn{3}{c}{Schools and Libraries (Miles)} & \multicolumn{3}{c}{\begin{tabular}[c]{@{}c@{}}All data points (Miles)\end{tabular}} \\ \cmidrule(l){2-4} \cmidrule(l){5-7}
\multicolumn{1}{c}{}    & Mean      & Median      & Max        & Mean   & Median    & Max   \\
\midrule
White                       & 1.34      & 1.00        & 18.98      & 1.14   & 0.82    & 13.66  \\
Black or African American   & 1.05      & 0.77        & 18.96      & 0.90   & 0.65    & 11.46   \\
Undesignated                & 1.18      & 0.86        & 18.97      & 0.97   & 0.70    & 11.47  \\
Other                       & 1.07      & 0.80        & 18.94      & 0.85   & 0.64    & 10.86    \\
Asian                       & 0.90      & 0.73        & 18.63       & 0.68   & 0.55    & 9.39     \\
American Indian or Alaska Native    & 1.65      & 1.29        & 18.90      & 1.31   & 1.01      & 10.02     \\
Two or more races           & 1.01      & 0.75        & 9.52      & 0.80   & 0.60  & 9.69    \\
Native Hawaiian  or Pacific Islander    & 1.20      & 0.77        & 6.43       & 0.82   & 0.58       & 4.26    \\ \bottomrule
\end{tabular}
\end{table}

\subsubsection{Normalized distances}
Tables \ref{tab:fl_balanced_fair_dist_normalized_1}, \ref{tab:nc_balanced_fair_dist_normalized_1}, \ref{tab:fl_balanced_fair_dist_normalized_5}, \ref{tab:nc_balanced_fair_dist_normalized_5}, \ref{tab:fl_balanced_fair_dist_normalized_9} and \ref{tab:nc_balanced_fair_dist_normalized_9} summarize the normalized median distances for the assignment by $L$-balanced fair $k$-median algorithm for $\epsilon = 0.1$, $\epsilon = 0.5$ and $\epsilon = 0.9$. Columns ``Schools and Libraries'' and ``All data points'' indicate the valid polling locations considered in the solution, and sub-columns ``Nearest school/library'' and ``Median pairwise voter distance'' indicate the normalization baseline.

\begin{table}[H]
\caption{Florida: Median normalized distances to polling locations in $L$-Balanced Fair $k$-median assignment for $\epsilon = 0.1$}
\label{tab:fl_balanced_fair_dist_normalized_1}
\centering
\resizebox{\textwidth}{!}{
\begin{tabular}{@{}lcccccc@{}}
\toprule
\multicolumn{1}{c}{\multirow{3}{*}{Race}} & \multicolumn{2}{c}{Schools and Libraries} & \multicolumn{2}{c}{\begin{tabular}[c]{@{}c@{}}All data points\end{tabular}} \\ \cmidrule(l){2-3} \cmidrule(l){4-5}
\multicolumn{1}{c}{}    & Nearest school/library      & Median pairwise voter distance        & Nearest school/library      & Median pairwise voter distance   \\
\midrule
White, not Hispanic         & 1   & 0.82   & 0.71    & 0.67 \\
Hispanic                    & 1   & 0.82   & 0.76    & 0.67 \\
Black, not Hispanic         & 1   & 0.78   & 0.82    & 0.66 \\
Unknown                     & 1   & 0.81   & 0.75    & 0.65 \\
Asian or Pacific Islander   & 1   & 0.81   & 0.74    & 0.65 \\
Other                       & 1   & 0.81   & 0.73    & 0.66 \\
Multi-Racial                & 1   & 0.80   & 0.74    & 0.64 \\
American Indian             & 1   & 0.78   & 0.69    & 0.60  \\ \bottomrule
\end{tabular}
}
\end{table}

\begin{table}[H]
\caption{North Carolina: Median normalized distances to polling locations in $L$-Balanced Fair $k$-median assignment for $\epsilon = 0.1$}
\label{tab:nc_balanced_fair_dist_normalized_1}
\centering
\resizebox{\textwidth}{!}{
\begin{tabular}{@{}lcccccc@{}}
\toprule
\multicolumn{1}{c}{\multirow{3}{*}{Race}} & \multicolumn{2}{c}{Schools and Libraries} & \multicolumn{2}{c}{\begin{tabular}[c]{@{}c@{}}All data points\end{tabular}} \\ \cmidrule(l){2-3} \cmidrule(l){4-5}
\multicolumn{1}{c}{}    & Nearest school/library      & Median pairwise voter distance        & Nearest school/library      & Median pairwise voter distance   \\
\midrule
White                               & 1   & 0.83   & 0.78    & 0.68 \\
Black or African American           & 1   & 0.85   & 0.87    & 0.69 \\
Undesignated                        & 1   & 0.85   & 0.79    & 0.66 \\
Other                               & 1   & 0.87   & 0.78    & 0.65 \\
Asian                               & 1   & 0.90   & 0.72    & 0.64 \\
American Indian or Alaska Native    & 1   & 0.76   & 0.75    & 0.59 \\
Two or more races                   & 1   & 0.86   & 0.78    & 0.65 \\
Native Hawaiian  or Pacific Islander& 1   & 0.80   & 0.60    & 0.50  \\ \bottomrule
\end{tabular}
}
\end{table}

\begin{table}[H]
\caption{Florida: Median normalized distances to polling locations in $L$-Balanced Fair $k$-median assignment for $\epsilon = 0.5$}
\label{tab:fl_balanced_fair_dist_normalized_5}
\centering
\resizebox{\textwidth}{!}{
\begin{tabular}{@{}lcccccc@{}}
\toprule
\multicolumn{1}{c}{\multirow{3}{*}{Race}} & \multicolumn{2}{c}{Schools and Libraries} & \multicolumn{2}{c}{\begin{tabular}[c]{@{}c@{}}All data points\end{tabular}} \\ \cmidrule(l){2-3} \cmidrule(l){4-5}
\multicolumn{1}{c}{}    & Nearest school/library      & Median pairwise voter distance        & Nearest school/library      & Median pairwise voter distance   \\
\midrule
White, not Hispanic         & 1   & 0.78   & 0.71    & 0.67 \\
Hispanic                    & 1   & 0.78   & 0.76    & 0.67 \\
Black, not Hispanic         & 1   & 0.76   & 0.82    & 0.66 \\
Unknown                     & 1   & 0.77   & 0.75    & 0.65 \\
Asian or Pacific Islander   & 1   & 0.76   & 0.74    & 0.65 \\
Other                       & 1   & 0.77   & 0.73    & 0.66 \\
Multi-Racial                & 1   & 0.76   & 0.74    & 0.64 \\
American Indian             & 1   & 0.75   & 0.69    & 0.60  \\ \bottomrule
\end{tabular}
}
\end{table}

\begin{table}[H]
\caption{North Carolina: Median normalized distances to polling locations in $L$-Balanced Fair $k$-median assignment for $\epsilon = 0.5$}
\label{tab:nc_balanced_fair_dist_normalized_5}
\centering
\resizebox{\textwidth}{!}{
\begin{tabular}{@{}lcccccc@{}}
\toprule
\multicolumn{1}{c}{\multirow{3}{*}{Race}} & \multicolumn{2}{c}{Schools and Libraries} & \multicolumn{2}{c}{\begin{tabular}[c]{@{}c@{}}All data points\end{tabular}} \\ \cmidrule(l){2-3} \cmidrule(l){4-5}
\multicolumn{1}{c}{}    & Nearest school/library      & Median pairwise voter distance        & Nearest school/library      & Median pairwise voter distance   \\
\midrule
White                               & 1   & 0.80   & 0.78    & 0.68 \\
Black or African American           & 1   & 0.82   & 0.87    & 0.69 \\
Undesignated                        & 1   & 0.80   & 0.79    & 0.66 \\
Other                               & 1   & 0.83   & 0.78    & 0.65 \\
Asian                               & 1   & 0.86   & 0.72    & 0.64 \\
American Indian or Alaska Native    & 1   & 0.75   & 0.75    & 0.59 \\
Two or more races                   & 1   & 0.83   & 0.78    & 0.65 \\
Native Hawaiian  or Pacific Islander& 1   & 0.75   & 0.60    & 0.50  \\ \bottomrule
\end{tabular}
}
\end{table}

\begin{table}[H]
\caption{Florida: Median normalized distances to polling locations in $L$-Balanced Fair $k$-median assignment for $\epsilon = 0.9$}
\label{tab:fl_balanced_fair_dist_normalized_9}
\centering
\resizebox{\textwidth}{!}{
\begin{tabular}{@{}lcccccc@{}}
\toprule
\multicolumn{1}{c}{\multirow{3}{*}{Race}} & \multicolumn{2}{c}{Schools and Libraries} & \multicolumn{2}{c}{\begin{tabular}[c]{@{}c@{}}All data points\end{tabular}} \\ \cmidrule(l){2-3} \cmidrule(l){4-5}
\multicolumn{1}{c}{}    & Nearest school/library      & Median pairwise voter distance        & Nearest school/library      & Median pairwise voter distance   \\
\midrule
White, not Hispanic         & 1   & 0.75   & 0.71    & 0.67 \\
Hispanic                    & 1   & 0.76   & 0.76    & 0.67 \\
Black, not Hispanic         & 1   & 0.82   & 0.82    & 0.66 \\
Unknown                     & 1   & 0.75   & 0.75    & 0.65 \\
Asian or Pacific Islander   & 1   & 0.74   & 0.74    & 0.65 \\
Other                       & 1   & 0.73   & 0.73    & 0.66 \\
Multi-Racial                & 1   & 0.74   & 0.74    & 0.64 \\
American Indian             & 1   & 0.69   & 0.69    & 0.60  \\ \bottomrule
\end{tabular}
}
\end{table}

\begin{table}[H]
\caption{North Carolina: Median normalized distances to polling locations in $L$-Balanced Fair $k$-median assignment for $\epsilon = 0.9$}
\label{tab:nc_balanced_fair_dist_normalized_9}
\centering
\resizebox{\textwidth}{!}{
\begin{tabular}{@{}lcccccc@{}}
\toprule
\multicolumn{1}{c}{\multirow{3}{*}{Race}} & \multicolumn{2}{c}{Schools and Libraries} & \multicolumn{2}{c}{\begin{tabular}[c]{@{}c@{}}All data points\end{tabular}} \\ \cmidrule(l){2-3} \cmidrule(l){4-5}
\multicolumn{1}{c}{}    & Nearest school/library      & Median pairwise voter distance        & Nearest school/library      & Median pairwise voter distance   \\
\midrule
White                               & 1   & 0.80   & 0.78    & 0.68 \\
Black or African American           & 1   & 0.82   & 0.87    & 0.69 \\
Undesignated                        & 1   & 0.81   & 0.79    & 0.66 \\
Other                               & 1   & 0.82   & 0.78    & 0.65 \\
Asian                               & 1   & 0.83   & 0.72    & 0.64 \\
American Indian or Alaska Native    & 1   & 0.75   & 0.75    & 0.59 \\
Two or more races                   & 1   & 0.82   & 0.78    & 0.65 \\
Native Hawaiian  or Pacific Islander& 1   & 0.74   & 0.60    & 0.50  \\ \bottomrule
\end{tabular}
}
\end{table}

\subsection{$k$-center}
\label{app:kcenter}
Tables~\ref{tab:fl_kcenter} and \ref{tab:nc_kcenter} summarize the normalized median distances for the assignments produced by the unconstrained $k$-center algorithm.

\begin{table}[H]
\caption{Florida: distance to polling places selected by unconstrained $k$-center}
\label{tab:fl_kcenter}
\centering
\begin{tabular}{@{}lccc@{}}
\toprule
Race & \begin{tabular}[c]{@{}c@{}}Mean distance\\ (Miles)\end{tabular} & \begin{tabular}[c]{@{}c@{}}Median distance\\ (Miles)\end{tabular} & \begin{tabular}[c]{@{}c@{}}Maximum distance\\ (Miles)\end{tabular} \\ \midrule
White, not Hispanic & 1.21 & 0.94 & 16.27 \\
Hispanic & 0.79 & 0.62 & 11.17 \\
Black, not Hispanic & 0.92 & 0.71 & 13.14 \\
Unknown & 0.89 & 0.71 & 10.02 \\
Asian or Pacific Islander & 0.95 & 0.77 & 10.56 \\
Other & 0.97 & 0.74 & 10.89 \\
Multi-Racial & 0.97 & 0.76 & 9.77 \\
American Indian & 1.22 & 0.90 & 11.73 \\ \bottomrule
\end{tabular}
\end{table}

\begin{table}[H]
\caption{North Carolina: distance to polling places selected by unconstrained $k$-center}
\label{tab:nc_kcenter}
\centering
\begin{tabular}{@{}lccc@{}}
\toprule
Race & \begin{tabular}[c]{@{}c@{}}Mean distance\\ (Miles)\end{tabular} & \begin{tabular}[c]{@{}c@{}}Median distance\\ (Miles)\end{tabular} & \begin{tabular}[c]{@{}c@{}}Maximum distance\\ (Miles)\end{tabular} \\ \midrule
White & 1.76 & 1.35 & 12.18 \\
Black or African American & 1.55 & 1.11 & 11.29 \\
Undesignated & 1.60 & 1.19 & 11.66 \\
Other & 1.46 & 1.09 & 12.25 \\
Asian & 1.18 & 0.95 & 12.41 \\
American Indian or Alaska Native & 2.46 & 2.20 & 11.02 \\
Two or more races & 1..44 & 1.08 & 11.08 \\
Native Hawaiian or Pacific Islander & 1.78 & 1.29 & 10.01 \\ \bottomrule
\end{tabular}
\end{table}

\subsection{Balanced $k$-center}

\subsubsection{Absolute distances}
Tables \ref{tab:fl_kcenter_balanced} and \ref{tab:nc_kcenter_balanced} summarize the normalized median distances for the assignments produced by the balanced $k$-center algorithm.

\begin{table}[H]
\caption{Florida: distance to polling places selected by balanced $k$-center algorithm, for three different capacities. }
\label{tab:fl_kcenter_balanced}
\begin{tabular}{@{}lccccccccc@{}}
\toprule
\multirow{2}{*}{Race} & \multicolumn{3}{c}{$\epsilon = 0.1$} & \multicolumn{3}{c}{$\epsilon = 0.5$} & \multicolumn{3}{c}{$\epsilon = 0.9$} \\ \cmidrule(l){2-4} \cmidrule(l){5-7} \cmidrule(l){8-10} 
 & \begin{tabular}[c]{@{}c@{}}Mean\\ (Miles)\end{tabular} & \begin{tabular}[c]{@{}c@{}}Medium\\ (Miles)\end{tabular} & \begin{tabular}[c]{@{}c@{}}Max\\ (Miles)\end{tabular} & \begin{tabular}[c]{@{}c@{}}Mean\\ (Miles)\end{tabular} & \begin{tabular}[c]{@{}c@{}}Medium\\ (Miles)\end{tabular} & \begin{tabular}[c]{@{}c@{}}Max\\ (Miles)\end{tabular} & \begin{tabular}[c]{@{}c@{}}Mean\\ (Miles)\end{tabular} & \begin{tabular}[c]{@{}c@{}}Medium\\ (Miles)\end{tabular} & \begin{tabular}[c]{@{}c@{}}Max\\ (Miles)\end{tabular} \\ \cmidrule(r){1-10}
White, not Hispanic & 3.62 & 1.48 & 50.30 & 2.24 & 1.20 & 27.33 & 1.68 & 1.08 & 19.14 \\
Hispanic & 2.54 & 0.89 & 50.26 & 1.50 & 0.73 & 27.25 & 1.10 & 0.68 & 17.22 \\
Black, not Hispanic & 3.19 & 1.10 & 50.17 & 1.92 & 0.87 & 27.30 & 1.39 & 0.80 & 17.62 \\
Unknown & 3.00 & 1.02 & 50.13 & 1.77 & 0.85 & 27.30 & 1.28 & 0.77 & 17.68 \\
Asian or Pacific Islander & 3.21 & 1.14 & 46.96 & 1.91 & 0.95 & 27.10 & 1.37 & 0.87 & 16.78 \\
Other & 3.08 & 1.11 & 47.01 & 1.87 & 0.91 & 27.30 & 1.37 & 0.83 & 16.51 \\
Multi-Racial & 3.24 & 1.17 & 47.01 & 1.97 & 0.96 & 27.07 & 1.42 & 0.87 & 18.50 \\
American Indian & 3.72 & 1.49 & 49.99 & 2.31 & 1.18 & 27.35 & 1.71 & 1.06 & 18.50 \\ \bottomrule
\end{tabular}
\end{table}

\begin{table}[H]
\caption{North Carolina: distance to polling places selected by balanced $k$-center algorithm, for three different capacities. }
\label{tab:nc_kcenter_balanced}
\resizebox{\textwidth}{!}{
\begin{tabular}{@{}lccccccccc@{}}
\toprule
\multirow{2}{*}{Race} & \multicolumn{3}{c}{$\epsilon = 0.1$} & \multicolumn{3}{c}{$\epsilon = 0.5$} & \multicolumn{3}{c}{$\epsilon = 0.9$} \\ \cmidrule(l){2-4} \cmidrule(l){5-7} \cmidrule(l){8-10} 
 & \begin{tabular}[c]{@{}c@{}}Mean\\ (Miles)\end{tabular} & \begin{tabular}[c]{@{}c@{}}Medium\\ (Miles)\end{tabular} & \begin{tabular}[c]{@{}c@{}}Max\\ (Miles)\end{tabular} & \begin{tabular}[c]{@{}c@{}}Mean\\ (Miles)\end{tabular} & \begin{tabular}[c]{@{}c@{}}Medium\\ (Miles)\end{tabular} & \begin{tabular}[c]{@{}c@{}}Max\\ (Miles)\end{tabular} & \begin{tabular}[c]{@{}c@{}}Mean\\ (Miles)\end{tabular} & \begin{tabular}[c]{@{}c@{}}Medium\\ (Miles)\end{tabular} & \begin{tabular}[c]{@{}c@{}}Max\\ (Miles)\end{tabular} \\ \cmidrule(r){1-10}
White & 3.49 & 2.09 & 39.43 & 2.48 & 1.74 & 26.80 & 2.14 & 1.57 & 18.71 \\
Black or African American & 3.46 & 1.82 & 38.84 & 2.41 & 1.46 & 26.80 & 2.01 & 1.31 & 15.73 \\
Undesignated & 3.38 & 1.85 & 39.36 & 2.35 & 1.52 & 26.77 & 1.99 & 1.39 & 17.33 \\
Other & 3.29 & 1.66 & 39.15 & 2.25 & 1.38 & 26.77 & 1.88 & 1.26 & 16.84 \\
Asian & 2.98 & 1.45 & 38.95 & 1.93 & 1.22 & 16.71 & 1.59 & 1.11 & 15.61 \\
American Indian or Alaska Native & 4.28 & 2.77 & 39.19 & 3.24 & 2.58 & 26.77 & 2.86 & 2.37 & 17.33 \\
Two or more races & 3.39 & 1.67 & 36.78 & 2.27 & 1.36 & 25.39 & 1.89 & 1.24 & 16.44 \\
Native Hawaiian or Pacific Islander & 3.23 & 1.96 & 17.90 & 2.77 & 1.71 & 13.61 & 2.06 & 1.30 & 13.61 \\ \bottomrule
\end{tabular}
}
\end{table}

\subsubsection{Normalized distances}
Tables \ref{tab:fl_normalized_kcenter_balanced} and \ref{tab:nc_normalized_kcenter_balanced} summarize the normalized median distances for the assignments produced by the balanced $k$-center algorithm.

\begin{table}[H]
\caption{Florida: Normalized distances to polling places selected by balanced $k$-center algorithm, for three different capacities.}
\label{tab:fl_normalized_kcenter_balanced}
\begin{tabular}{@{}lccccccccc@{}}
\toprule
\multirow{2}{*}{Race} & \multicolumn{3}{c}{Nearest School/Library} & \multicolumn{3}{c}{Median pairwise voter distance} \\ \cmidrule(l){2-4} \cmidrule(l){5-7} 
 & \begin{tabular}[c]{@{}c@{}}$\epsilon = 0.1$\end{tabular} & \begin{tabular}[c]{@{}c@{}}$\epsilon = 0.5$\end{tabular} &  \begin{tabular}[c]{@{}c@{}}$\epsilon = 0.9$\end{tabular} & \begin{tabular}[c]{@{}c@{}}$\epsilon = 0.1$\end{tabular} &  \begin{tabular}[c]{@{}c@{}}$\epsilon = 0.5$\end{tabular} & \begin{tabular}[c]{@{}c@{}}$\epsilon = 0.9$\end{tabular} \\ \cmidrule(r){1-7}
White, not Hispanic & 1.79 & 1.42 & 1.29 & 1.56 & 1.30 & 1.19 \\
Hispanic & 1.89 & 1.50 & 1.36 & 1.51 & 1.26 & 1.16\\
Black, not Hispanic  & 2.30 & 1.80 & 1.61 & 1.71 & 1.40 & 1.28\\
Unknown & 1.87 & 1.50 & 1.36 & 1.51 & 1.26 & 1.16\\
Asian or Pacific Islander & 1.87 & 1.48 & 1.34 & 1.46 & 1.24 & 1.14\\
Other & 1.83 & 1.45 & 1.33 & 1.50 & 1.25 & 1.15\\
Multi-Racial & 2.05 & 1.61 & 1.44 & 1.59 & 1.33 & 1.21\\
American Indian & 2.01 & 1.59 & 1.43 & 1.62 & 1.34 & 1.23\\ \bottomrule
\end{tabular}
\end{table}

\begin{table}[H]
\caption{North Carolina: Normalized distances to polling places selected by balanced $k$-center algorithm, for three different capacities.}
\label{tab:nc_normalized_kcenter_balanced}
\begin{tabular}{@{}lccccccccc@{}}
\toprule
\multirow{2}{*}{Race} & \multicolumn{3}{c}{Nearest School/Library} & \multicolumn{3}{c}{Median pairwise voter distance} \\ \cmidrule(l){2-4} \cmidrule(l){5-7} 
 & \begin{tabular}[c]{@{}c@{}}$\epsilon = 0.1$\end{tabular} & \begin{tabular}[c]{@{}c@{}}$\epsilon = 0.5$\end{tabular} &  \begin{tabular}[c]{@{}c@{}}$\epsilon = 0.9$\end{tabular} & \begin{tabular}[c]{@{}c@{}}$\epsilon = 0.1$\end{tabular} &  \begin{tabular}[c]{@{}c@{}}$\epsilon = 0.5$\end{tabular} & \begin{tabular}[c]{@{}c@{}}$\epsilon = 0.9$\end{tabular} \\ \cmidrule(r){1-7}
White & 1.66 & 1.41 & 1.30 & 1.42 & 1.23 & 1.14 \\
Black or African American & 2.15 & 1.79 & 1.57 & 1.62 & 1.39 & 1.29\\
Undesignated & 1.83 & 1.52 & 1.40 & 1.47 & 1.28 & 1.19\\
Other & 1.89 & 1.56 & 1.43 & 1.49 & 1.30 & 1.20\\
Asian & 1.93 & 1.59 & 1.42 & 1.55 & 1.33 & 1.24\\
American Indian or Alaska Native & 1.74 & 1.56 & 1.37 & 1.42 & 1.30 & 1.19\\
Two or more races & 2.04 & 1.65 & 1.52 & 1.57 & 1.37 & 1.26\\
Native Hawaiian or Pacific Islander & 1.70 & 1.67 & 1.38 & 1.46 & 1.47 & 1.13\\ \bottomrule
\end{tabular}
\end{table}

\section{Polling place load in alternative voting location assignments}
Tables \ref{tab:fl_normalized_load} and \ref{tab:nc_normalized_load} summarize the normalized polling locations' loads for the assignments produced by the proposed fair algorithms.

\begin{table}[H]
\caption{Florida: normalized polling site load for different fair algorithms}
\label{tab:fl_normalized_load}
\centering
\begin{tabular}{@{}lccccccc@{}}
\toprule
Race & Fair $k$-median & \multicolumn{3}{c}{\begin{tabular}[c]{@{}c@{}} $L$-Balanced Fair $k$-median\end{tabular}}
& \multicolumn{3}{c}{\begin{tabular}[c]{@{}c@{}} Balanced $k$-center\end{tabular}}\\ \cmidrule(l){3-5} \cmidrule(l){6-8}
\multicolumn{1}{c}{} & & $\epsilon = 0.1$ & $\epsilon = 0.5$ & $\epsilon = 0.9$ & $\epsilon = 0.1$ & $\epsilon = 0.5$ & $\epsilon = 0.9$\\\midrule
White, not Hispanic & 1 & 1 & 1 & 1 & 1 & 1 & 1\\
Hispanic & 0.91 & 0.96 & 0.93 & 0.92 & 0.85 & 0.81 & 0.85 \\
Black, not Hispanic & 0.86 & 0.93 & 0.92 & 0.89 & 0.88 & 0.95 & 0.88\\
Unknown & 0.91 & 0.97 & 0.96 & 0.94 & 0.89 & 0.94 & 0.86\\
Asian or Pacific Islander & 0.97 & 1.01 & 1.01 & 1.02 & 0.97 & 0.97 & 0.99\\
Other & 0.96 & 1 & 0.99 & 0.99 & 0.91 & 0.95 & 0.95\\
Multi-Racial & 0.95 & 0.99 & 0.99 & 0.98 & 0.93 & 0.96 & 0.96 \\
American Indian & 0.94 & 0.96 & 0.96 & 0.95 & 0.94 & 0.97 & 0.99\\ \bottomrule
\end{tabular}
\end{table}

\begin{table}[H]
\caption{North Carolina: normalized polling site load for different fair algorithms}
\label{tab:nc_normalized_load}
\centering
\begin{tabular}{@{}lccccccc@{}}
\toprule
Race & Fair $k$-median & \multicolumn{3}{c}{\begin{tabular}[c]{@{}c@{}} $L$-Balanced Fair $k$-median\end{tabular}}
& \multicolumn{3}{c}{\begin{tabular}[c]{@{}c@{}} Balanced $k$-center\end{tabular}}\\ \cmidrule(l){3-5} \cmidrule(l){6-8}
\multicolumn{1}{c}{} & & $\epsilon = 0.1$ & $\epsilon = 0.5$ & $\epsilon = 0.9$ & $\epsilon = 0.1$ & $\epsilon = 0.5$ & $\epsilon = 0.9$\\\midrule
White & 1 & 1 & 1 & 1 & 1 & 1 & 1\\
Black or African American & 1.12 & 1.02 & 1.05 & 1.07 & 1.02 & 1.03 & 1.07\\
Undesignated & 1.08 & 1.04 & 1.05 & 1.07 & 1.07 & 1.05 & 1.08\\
Other & 1.12 & 1.07 & 1.11 & 1.12 & 1.17 & 1.09 & 1.15\\
Asian & 1.12 & 1.10 & 1.15 & 1.18 &  1.25 & 1.23 & 1.21\\
American Indian or Alaska Native & 0.74 & 0.79 & 0.77 & 0.71 & 0.62 & 0.66 & 0.71\\
Two or more races & 1.13 & 1.06 & 1.08 & 1.11 & 1.13 & 1.09 & 1.14 \\
Native Hawaiian or Pacific Islander & 1.08 & 1.02 & 1.06 & 1.05 & 1.01 & 1.05 & 1.13\\ \bottomrule
\end{tabular}
\end{table}

\end{document}

%% file: intro.tex
\section{Introduction}
\label{polling:intro}

Convenient access to voting is a crucial component of elections. Prior studies have revealed that even after controlling for the variables that account for people's  motivation to vote, access to polling locations has a significant effect on voter turnout~\cite{gimpel2003political}. Long wait times can discourage voting or even cause voters to abandon lines at the polling sites, reduce voters' confidence in their votes being counted or impose financial costs on them~\cite{stewart2015waiting,stein2020waiting}. If voters have to travel long distances to reach a polling location, it  heavily influences how they choose to vote or might cause them to reconsider voting altogether, even if they are already registered to vote~\cite{dyck2005distance}. For example, for many voters it is only feasible to cast a ballot either before or after work. Thus, the ability to access a polling place within reach in these small time windows is the deciding factor on whether a voter casts their vote or not.

Limitations on access to voting are also refracted through a demographic lens. Several studies have identified clear demographic disparities with respect to both distance to polling sites and wait times at said locations~\cite{chen2019racial,joslyn2020distance,stewart2015waiting,dyck2005distance}. For example, in the 2016 US presidential election, voters in predominantly black neighborhoods waited $29\%$ longer at polling locations than those in white neighborhoods~\cite{chen2019racial}. The significant effect that both travel distance and wait times have on voter turnout coupled with their observed disparate effect on different demographics raise serious concerns about whether an election truly is representative of the preferences of all communities.

Quantifying voting access disparities is the first step towards improving access to voting. With a clear measure of access disparity in place, it is then possible to explore different mechanisms to reduce disparities and improve access, whether it be relocating polling locations or adding more locations and options to vote. However, even the simplest measure of access to voting -- the time taken to travel to a polling location -- does not lend itself to an easy measurement of disparities because of how different the living and travel environment might be even from county to county within a state. For example in a metropolitan area with high population density might admit shorter times to access a polling location than a more rural area with a lower population density. It might further be the case that the metropolitan area is dominated by members of one community, and the rural area by members of a different community. But this does not necessarily that the resultant disparity in access times is evidence of overt voter suppression or even some form of structural bias. Rather, this disparity needs to be calibrated against community norms for reasonable travel times, potentially also adjusted against population density that might manifest in variable queueing times at polling locations.

\paragraph{Our Contributions.}

In this paper we develop a novel methodology to quantify and calibrate disparities in access to voting using measures that act as proxies for the time taken to travel to a polling station and the waiting time to vote. We further develop scalable algorithms that can redistribute polling locations in a number of different ways so as to improve access disparities. This investigation is conducted within a specific context: that of voter access disparities in Florida and North Carolina during the 2020 general election. The main findings of our investigation can be summarized as follows:

\begin{itemize}
\item We found that when ``normalized'' distances are considered non-white voters in general had to travel further to the polls in Florida while in North Carolina any racial disparity in travel was smaller.
\item There was some racial variation in ``load'' (our proxy for the waiting time to vote) but this did not suggest any substantive impact. 
\item Our algorithmic interventions (both to minimize travel time as well as waiting time) improved access disparities across the board, while making visible the tradeoff between disparity mitigation and the resources required (in the form of additional polling locations).
\end{itemize}

In section~\ref{sec:data} we discuss our approach in collecting and pre-processing the data needed in our analysis. In  section~\ref{sec:polling-stats} we adopt the two measures proposed above to analyze and compare different racial groups in terms of their access to polling locations in Florida and North Carolina. In section~\ref{sec:sol} we propose three methods that produce fairer alternatives to the polling location selection with respect to the proposed measures, and assess these methods on the Florida and North Carolina case studies in Section \ref{sec:experiments}.  Finally, in Section \ref{sec:conclusions}, we discuss limitations and conclude. 

\subsection{Related Work}
There exist a rich literature on polling location placement and voting. Some of these works are closely related to ours. For example, some have studied how changes to precincts and polling places on Election Day affect voter turnout and whether these changes target voters based on their partisanship or race~\cite{shepherd2021politics,clinton2021polling,yoder2018polling}. And a number of other works have focused on the effects of distance to the polling locations~\cite{cantoni2020precinct,bagwe2020polling,tomkins2021blocks} and wait times at these locations~\cite{pettigrew2021downstream} on voter turnout.

The main algorithmic tools we use are drawn from the literature on \emph{fair clustering} \cite{makarychev2021approximation,chierichetti2017fair,kleindessner2019fair,kleindessner2019guarantees,chen2019proportionally,schmidt2018fair,ahmadian2019clustering,rosner2018privacy,vakilian2021improved,abraham2019fairness,negahbani2021better,Chlamatac2022Approximating}, a sub-field of algorithmic fairness~\cite{dwork2012fairness,feldman2015certifying,hardt2016equality,kusner2017counterfactual,zafar2017fairness,abbasi2019fairness} which has gained a lot of attention in recent years. Fair clustering has often been used as a tool for redistricting\cite{fryer2011measuring,mehrotra1998optimization, levin2019automated} and even fairness in redistricting \cite{stoica2019minimizing}, but there has been little to no research on making \emph{access} to voting more equitable. 
The closest related work investigates equitable service to customers in the context of facility location (e.g. \cite{marsh1994equity,drezner2011note,drezner2007equity}); however,  previously proposed solutions often are either tailored to specific problems or cannot scale to the data sizes we encounter in our problem.


%% file: data.tex
\section{Data collection and preprocessing}
\label{sec:data}

In order to conduct our case studies of voting access disparities, there are two key pieces of data we need to collect: voter registration information that includes race and voter residential addresses, and associated polling location information (Section \ref{sec:data_voter_polling}).  In order to perform disparity analyses, we also need latitude and longitude data for each voter and polling location address, as well as latitude and longitude data for the alternative or new polling locations that may be introduced (school and library locations).  The determination of an address' latitude and longitude is known as \emph{geocoding} (Section \ref{sec:data_geocoding}).  Although these data are publicly available, collecting geocoding information is difficult and thus the final data collected presents additional privacy risks to voters beyond what is already publicly available; we have chosen not to make this final collected data public.

\subsection{Voting Rolls and Polling Location Data}
\label{sec:data_voter_polling}
Voting rolls, or voter registration data, usually include the name, address, other voting-related information, and race for all registered voters.  Voting rolls are publicly available for most states in the U.S., though some states charge for access to this data, require the requester to be a state resident, or don't publicly release information about voter race.  Florida and North Carolina were chosen as focus states for our case study because the states are large, have substantial non-white populations, and provide freely available voting rolls that include race information.\footnote{Florida voting rolls are available at: \url{http://flvoters.com/downloads.html}.  Our analyses use the 10/13/2020 data.  North Carolina voting rolls are available at: \url{https://www.ncsbe.gov/results-data/voter-registration-data}. Our analyses use the 11/03/2020 snapshot.} There are approximately 15.1 and 7.3 million voters in the records for Florida and North Carolina, respectively.  Data collected for both states for each voter included: voter's residential address; county, precinct, and congressional district information to determine the voter's associated precinct location; and race as defined in the voting rolls.

In the United States, a \emph{precinct} is the smallest unit electoral districts are divided into by the local government. Each precinct has one polling place designated to it on election day where voters cast their ballots, and more than one precinct may use the same polling place. Voters' residential addresses determine their assigned precinct. Both Florida and North Carolina require their residents to vote at the polling place designated to their precinct. We use the terms polling location and precinct site interchangeably to refer to designated polling places from the 2020 general election.  We collected all polling location addresses for Florida and North Carolina\footnote{Florida polling locations were downloaded from \url{https://dos.myflorida.com/elections/for-voters/voting/}. North Carolina polling locations were downloaded from: \url{https://www.ncsbe.gov/results-data/polling-place-data}.  Our analyses are based on the November 2020 general election polling locations.} and determined each voter's polling location based on their listed county and precinct information from the voting rolls. 6068 and 2662 polling sites are listed in the Florida and North Carolina records, respectively.

In Section \ref{sec:sol} we will consider alternative sites that could serve as polling locations.  In order to do these analyses, we collected the addresses of all schools and libraries in Florida and North Carolina.\footnote{Florida school addresses were downloaded from \url{https://web03.fldoe.org/Schools/schoolreport.asp} and library addresses from \url{https://publiclibraries.com/state/florida/}.  North Carolina school addresses were downloaded following the instructions at \url{https://files.nc.gov/dpi/documents/fbs/accounting/eddie/createreport.pdf} and library addresses from \url{https://statelibrary.ncdcr.gov/ld/about-libraries/library-directory/download}.} In total, the addresses for 5128 and 3265 schools and libraries were collected in Florida and North Carolina, respectively.  Schools and libraries were chosen since these are public sites that are generally well-distributed around a state, with each community having access to a school and a library, and  are often sites chosen to serve as polling locations.

\subsection{Geocoding Addresses}
\label{sec:data_geocoding}
One aspect of voting access disparities that we will consider is voters' distance to the nearest polling location.  In order to determine that, we need the latitude and longitude for each voter's residence, each polling location, and each school or library that could serve as a possible polling location.  The process of determining the latitude and longitude of the addresses for each of these sites (collected as described in Section \ref{sec:data_voter_polling}) is known as \emph{geocoding}.

We used the ArcGIS tool \footnote{\url{https://desktop.arcgis.coml}} for our geocoding queries.  In the ArcGIS geocoder, a \emph{match score} controls how closely addresses have to match their most likely candidate in the reference data.  The minimum match score is set to the default value ($80$), as are all other internal parameters. The match score determines how similar records must be to be declared a match. When geo-locating an address, the address is compared to records in a database. Each of these records has an associated latitude and longitude. The record with the highest match score with that address is chosen as the match and that record's associated latitude and longitude are assigned to the address. If there are no records with a match score above the minimum match score, no match is returned. By using the default minimum match score of 80, we allow for some spelling mistakes and other slight address differences between the address and its match, while maintaining a high degree of match accuracy.


\subsection{Missing Data Analysis}
In order to ensure that no eligible voters were missing from our voting rolls, we chose the latest possible voting roll snapshot that occurred before the 2020 general election.  Since neither state allows same-day voter registration on election day, this includes all voters who were eligible for the election.\footnote{North Carolina allows same-day voter registration during the early voting period; these voters were included in the 11/3/2020 voting rolls snapshot we used since 11/3/2020 was the day of the general election.  Florida allows voter registration until 29 days before the election, so all eligible voters would have been listed in the 10/13/2020 snapshot we used.}  Our lists of valid polling locations, schools, and libraries are also complete.\footnote{To the best of our knowledge, the sources we used to gather addresses for polling locations, schools, and libraries in each state provide a complete list of all such locations.}

However, the geocoding process did not successfully return latitude and longitude matches for all given addresses; in some cases, the provided address was misspelled, incomplete, or potentially missing from the ArcGIS database.  
Within each state, approximately 91 percent of voter addresses were successfully geocoded. Out of the 6068 polling locations in Florida, 5159 were successfully geocoded, and in North Carolina 2250 out of the 2662 polling locations were successfully geocoded. Out of all the collected addresses for schools and libraries within each state, 4499 and 2711 were successfully geocoded and used in the analyses for Florida and North Carolina, respectively.

Voters with residential addresses or polling locations that could not be geocoded were excluded from the data used for analysis in the remainder of this paper, as were polling locations that could not be geocoded. Specifically, in all the following analyses 11.8 and  5.3 million voter records were included for Florida and North Carolina, respectively. In order to ensure this missing data did not skew our voting access disparity analysis, we conducted a missing data analysis. The geocoding success rates per race are given in Figure \ref{fig:geocode_hitrate}.  The middle column in both tables presents the fraction of voters within each racial group whose address was successfully geocoded, and the last column shows the fraction of voters within each group with a successfully geocoded precinct location. Figure \ref{fig:voter_dist} compares the racial distribution of voter records used in our analysis to that of the original data for Florida and North Carolina, respectively. The results indicate that no group is over- or under-represented as the result of the geolocating process. We do not include confidence intervals on summary statistics. Confidence intervals express uncertainty due to sampling from a (sometimes hypothetical) population. In our case, at least for the purposes of Figure \ref{fig:geocode_hitrate}, our data includes everyone in the population of interest-- registered voters in their state. Thus, the empirical rate shown in these tables is-- by definition-- the ``true" value of the parameter for this population. 

\begin{figure}[htbp]
    \centering
    \begin{tabular}{cc}
    \textbf{Geocoding success rate in Florida} &
    \textbf{Geocoding success rate in North Carolina}\\

\begin{tabular}{lcc}
\hline
                         Race & \begin{tabular}[c]{@{}c@{}}Geocoded \\ Voters \end{tabular} & \begin{tabular}[c]{@{}c@{}} Voters w/ geocoded \\ polling loc.s \end{tabular} \\ \hline
White, not Hisp.       & 0.91   & 0.83                                                                    \\
Hispanic                  & 0.91   & 0.85                                                                    \\
Black, not Hisp.       & 0.92   & 0.88                                                                   \\
Unknown                   & 0.92    & 0.86                                                                   \\
Asian or Pac. Isl. & 0.89    & 0.85                                                                   \\
Other                     & 0.91    & 0.85                                                                   \\
Multi-Racial              & 0.92    & 0.85                                                                   \\
American Indian           & 0.92    & 0.84                                                                   \\ \hline
\end{tabular}

&

\begin{tabular}{lcc}
\hline
                         Race & \begin{tabular}[c]{@{}c@{}}Geocoded \\ Voters\end{tabular} & \begin{tabular}[c]{@{}c@{}} Voters w/ geocoded \\ polling loc.s\end{tabular} \\ \hline
White & 0.91 & 0.78 \\
Black / Afr. Amer. & 0.91 & 0.77\\
Undesignated & 0.87 &  0.77\\
Other & 0.90 & 0.78\\
Asian & 0.85 & 0.81\\
Amer. Ind. / AK Nat. & 0.86 & 0.71\\
Two or more races & 0.89 & 0.76\\
Nat. Haw. / Pac. Isl. & 0.87 & 0.79\\ \hline
\end{tabular}

    \end{tabular}
    \caption{Geocoding success rates for Florida (left) and North Carlolina (right) broken down by race.  Tables show fraction of voters with addresses correctly geocoded as well as fraction of voters whose polling location is correctly geocoded.}
    \label{fig:geocode_hitrate}
\end{figure}

\begin{figure}[htbp]
    \centering
    \begin{tabular}{cc}
    \textbf{Voter distribution by race in Florida} &
    \textbf{Voter distribution by race in North Carolina}\\

\begin{tabular}{lcc}
\hline
                        Race  & \begin{tabular}[c]{@{}c@{}}Original Data \end{tabular} &
                        \begin{tabular}[c]{@{}c@{}}Analysis Data \end{tabular}
                        \\ \hline
White, not Hisp.       & 0.616   & 0.605                                                                \\
Hispanic                  & 0.171   & 0.174                                                                \\
Black, not Hisp.       & 0.135   & 0.141                                                                \\
Unknown                   & 0.029   & 0.030                                                                \\
Asian or Pac. Isl. & 0.020   & 0.019                                                                \\
Other                     & 0.017   & 0.017                                                                \\
Multi-Racial              & 0.006   & 0.006                                                                \\
American Indian           & 0.003   & 0.003                                                                \\ \hline
\end{tabular}

&

\begin{tabular}{lcc}
\hline
 Race & \multicolumn{1}{c}{\begin{tabular}[c]{@{}c@{}}Original Data \end{tabular}} & \multicolumn{1}{c}{\begin{tabular}[c]{@{}c@{}} Analysis Data\end{tabular}} \\ \hline
White & 0.637 & 0.639 \\
Black / Afr. Amer. & 0.211 & 0.209\\
Undesignated & 0.097 &  0.097\\
Other & 0.027 & 0.027\\
Asian & 0.013 & 0.013\\
Amer. Ind. / AK Nat. & 0.007 & 0.006\\
Two or more races & 0.006 & 0.006\\
Nat. Haw.  / Pac. Isl. & 0.00004 & 0.00005\\ \hline
\end{tabular}
\end{tabular}
\caption{The voter distributions by race for Florida (left) and North Carolina (right) for both the original full voting rolls data and the analysis data after voters are removed to due unsuccessful geocoding.}
\label{fig:voter_dist}
\end{figure}


%% file: stats.tex
\section{Voting Access Disparities in Florida and North Carolina}
\label{sec:polling-stats}

In this section, we look into the accessibility of designated polling sites on 2020 Election Day to different race categories  defined in the Florida and North Carolina voter registration data. In our analysis we consider two measures of access to polling places:
\begin{description}
    \item[Distance:] we use the great-circle  distance, which is the shortest distance between two points on the surface of a sphere measured along its surface\footnote{\url{https://en.wikipedia.org/wiki/Great-circle_distance}}, as a proxy for the travel distance/time it takes for voters to get to their assigned polling site.
    \item[Polling site load:] we estimate the wait times at each polling place via the number of voters assigned to it. Our underlying assumption for this measure is that the number of votes that can be processed across all polling sites in a given time interval is the same.
\end{description}

These measures may not be perfect proxies for the \emph{time} it costs voters to participate in elections. For example, the geographical distance does not take into account the varying characteristics of travel routes in different areas, e.g. road traffic in rural vs urban regions. The disparate effect of photographic voter identification (ID) requirements on different demographics is disregarded by the polling site load measure \cite{stein2020waiting}. However, these measures provide us with an attainable yet useful enough tool to analyze and compare voters' access.

\subsection{Results: distances and polling site loads}
The results in Figure \ref{fig:raw_analysis} give the distance and polling site load for voters, broken down by race, in Florida and North Carolina.  The polling site load is measured based on the number of voters assigned to a voter's designated polling location.  That is, to determine the mean value for the number of voters per location for a specific race, for each of its members the number of voters at their respective polling site is included in the summation.

We observe that in Florida, White voters travel the longest distance to polling locations among the groups with respect to both mean and median values. Asian or Pacific Islander voters have the largest number of voters at their respective polling places. In North Carolina the American Indian voters experience the longest trip to polling locations while Asian voters encounter the most crowded polling locations. 
Prior studies have shown that variations in distance to polling sites, even as small as those observed in Figure \ref{fig:raw_analysis}, can have a significant effect on voter turnout \cite{haspel2005location}.

\begin{figure}[htbp]
    \centering
    \begin{tabular}{cc}
    \textbf{Florida} &
    \textbf{North Carolina}\\

\centering

\begin{tabular}{@{}lccccc@{}}
\toprule
\multicolumn{1}{c}{\multirow{3}{*}{Race}} & \multicolumn{3}{c}{Distance (miles)} & \multicolumn{2}{c}{\begin{tabular}[c]{@{}c@{}}Polling Load\\ (\# Voters)\end{tabular}} \\ \cmidrule(l){2-4} \cmidrule(l){5-6}
\multicolumn{1}{c}{}                      & Mean      & Med.      & Max        & Mean                          & Med.                           \\
\midrule
Wht. not Hisp. & 1.26 & 0.92 & 35.32 & 4805 & 4272 \\
Hispanic            & 0.89 & 0.68 & 28.51 & 4374 & 3824 \\
Blk. not Hisp. & 0.87 & 0.66 & 35.32 & 4160 & 3604 \\
Unknown             & 0.97 & 0.73 & 35.32 & 4460 & 3905 \\
Asn./Pac. Isl.  & 1.05 & 0.85  & 17.36 & 4903 & 4395 \\
Other               & 1.04 & 0.79 & 20.04 & 4758 & 4237 \\
Multi-Racial        & 1.02 & 0.78 & 16.11 & 4729 & 4216 \\
Amer. Ind.     & 1.24 & 0.86 & 22.96 & 4656 & 4079 \\ 
\midrule
Overall
& 1.12      & 0.82        & 35.26      & 4628                          & 4079                        
\\ 
\bottomrule
\end{tabular}

&
\begin{tabular}{@{}lccccc@{}}
\toprule
\multicolumn{1}{c}{\multirow{3}{*}{Race}} & \multicolumn{3}{c}{Distance (miles)} & \multicolumn{2}{c}{\begin{tabular}[c]{@{}c@{}}Polling Load\\ (\# Voters)\end{tabular}} \\ \cmidrule(l){2-4} \cmidrule(l){5-6}
\multicolumn{1}{c}{}    & Mean      & Med.      & Max        & Mean   & Med. \\
\midrule
White    & 1.69      & 1.27     & 40.67      & 3503      & 3105  \\
Black/Afr. Amer.   & 1.32      & 0.92        & 21.55      & 3581   & 3106 \\
Undesignated    & 1.51      & 1.10        & 25.68      & 3750   & 3302 \\
Other   & 1.40      & 1.03        & 16.87      & 3922   & 3418    \\
Asian   & 1.23      & 0.95        & 18.18      & 4204   & 3740     \\
Amer. Ind./AK Nat.    & 1.97      & 1.54        & 15.99      & 2751   & 2397  \\
Two or more races   & 1.32      & 0.94        & 16.06      & 3802   & 3327 \\
Nat. Haw./Pac. Isl.    & 1.65      & 1.16        & 9.37       & 3922   & 3326    \\ 
\midrule
Overall
& 1.58      & 1.16        & 40.62      & 3561                  & 3127 \\   
\bottomrule
\end{tabular}
\end{tabular}
\caption{Florida (left) and North Carolina (right) distance and polling site load values, broken down by race.  Distance is measured based on the mean, median, and maximum distance (in miles) for voters to their closest polling location.  Polling side load gives the mean and median values for number of voters assigned to each racial group's polling locations.  Maximum values are the same across all groups; the maximum size polling location for Florida is 18369 and for North Carolina is 13669.  The (non-weighted) numbers across all voters is given for all measures in the last row.}
\label{fig:raw_analysis}
\end{figure}

\subsection{Assessing disparities for those with the least access}
Next, we assess whether there are disparities among those with the furthest distance or largest polling sites.  These are the voters who, based on these measures, have the least access to their polling sites, so racial disparities among these voters is especially important to assess.  We do this by looking at the racial distributions and median distances and polling loads for the 10\% of voters with the least access, where access is defined as either distance or load. 

  To calculate these values, we first take the top 10\% of voters by distance (load) in the overall distribution. We then divide this sub-population by race and calculate the median distance (load) for each race group separately.  We also report the the proportion of voters in each race group (labeled ``Dist." for distribution) in the top 10\%.  The results are shown in Figure \ref{fig:worst10}. As expected, the median distance (given in the ``Miles" column) and median number of voters at their polling locations increase substantially for these voters relative to the overall state medians (Figure \ref{fig:raw_analysis}).  When comparing the racial distributions for the analysis data overall (Figure \ref{fig:voter_dist}) to the distribution by race of the voters with least access, we see that white voters make up a larger share of the top 10\% by distance than they do the overall population. The racial distribution of the top 10\% by polling site load is largely similar to that overall. While the top 10\% do exhibit a different racial composition than the overall population, the degree of difference does not suggest massive outlier effects that would necessitate more in depth analysis of this group separately.


\begin{figure}[htbp]
    \centering
    \begin{tabular}{cc}
    \textbf{Florida, worst 10\%} &
    \textbf{North Carolina, worst 10\%}\\


\begin{tabular}{@{}lcccc@{}}
\toprule
\multicolumn{1}{c}{\multirow{3}{*}{Race}} & \multicolumn{2}{c}{\begin{tabular}[c]{@{}c@{}} Distance\end{tabular}} & \multicolumn{2}{c}{\begin{tabular}[c]{@{}c@{}} Polling Load\end{tabular}}\\ 
\cmidrule(l){2-3}\cmidrule(l){4-5} 
\multicolumn{1}{c}{} & Dist. & Miles & Dist. & \#Voters\\
\cmidrule(l){1-5}
Wht. not Hisp. & 0.77 & 3.16 & 0.69 & 9637\\
Hispanic & 0.09 & 2.90 & 0.14 & 9637\\
Blk. not Hisp.  & 0.07 & 2.98 & 0.09 & 9923\\
Unknown & 0.02 & 2.97 & 0.02 & 9756\\
Asn./Pac. Isl. & 0.01 & 2.84 & 0.02 & 9930\\
Other & 0.01 & 3.00 & 0.01 & 9938\\
Multi-Racial & 0.004 & 2.98 & 0.006 & 9860\\
Amer. Ind. & 0.004 & 3.24 & 0.003 & 9756\\ \bottomrule
\end{tabular}

&

\begin{tabular}{@{}lccccc@{}}
\toprule
\multicolumn{1}{c}{\multirow{3}{*}{Race}} & \multicolumn{2}{c}{\begin{tabular}[c]{@{}c@{}} Distance\end{tabular}} & \multicolumn{2}{c}{\begin{tabular}[c]{@{}c@{}} Polling Load\end{tabular}}\\ 
\cmidrule(l){2-3}\cmidrule(l){4-5} 
\multicolumn{1}{c}{} & Dist. & Miles & Dist. & \#Voters\\
\cmidrule(l){1-5}
White & 0.76 & 4.31 & 0.58 & 7416\\
Black/Afr. Amer. & 0.14 & 4.38 & 0.23 & 7199\\
Undesignated & 0.08 & 4.29 & 0.11 & 7268\\
Other & 0.02 & 4.28 & 0.03 & 7532\\
Asian & 0.005 & 4.17 & 0.02 & 7169\\
Amer. Ind./AK Native & 0.01 & 4.25 & 0.002 & 7532\\
Two or more races & 0.003 & 4.25 & 0.007 & 7416\\
Nat. Haw. or Pac. Isl. & 0.00006 & 4.64 & 0.00004 & 9850\\ \bottomrule
\end{tabular}
\end{tabular}
\caption{The distance and polling site loads for the 10\% of voters with the worst access based on these measures in Florida and North Carolina.  The distribution of voters by race (left) and median measured values for these worst-off voters (right) are given.
}
\label{fig:worst10}
\end{figure}

\subsection{Normalized distances and polling site loads}
Gimpel \etal demonstrated that the burden of distance on voter turnout is experienced differently in different regions \cite{gimpel2003political}. In other words, the same distance might be perceived differently in rural as opposed to urban areas, due to factors such as differences in road traffic, differences in access to personal vehicles, or the norm for everyday travel distances. In order to determine whether the disparities we observed based on the measured distances to polling locations (in miles) are burdensome or otherwise unusual in voters' lives, we next explore two different methods for \emph{normalizing} these distances.  

\textbf{Distance Normalization: Schools and libraries.} In the first, we consider the distance to a polling location relative to other distances that voters might travel regularly and which might be considered reasonable.  For this analysis, we choose to use schools and libraries as reference locations, with the motivation that these sites are distributed so that all residents have access to a local school and library.  An important caveat is that this choice may hide voting access disparities that accrue to the same voters as school and library access disparities.  However, our additional analyses in Sections \ref{sec:sol} and \ref{sec:experiments} will consider the possibility of opening new polling locations, and since schools and libraries also serve as good polling locations, they provide for a useful analysis reference. 

\textbf{Distance Normalization: Median pairwise distance to other voters.}  Gimpel \etal showed that residential density acts as a barrier to voter turnout \cite{gimpel2003political}. In order to normalize based on residential density, we estimate how densely populated each voter's residence area is based on the median of their pairwise distances to other voters in their precinct as the baseline. For each voter, we find the median of their distances to all the other voters within the same precinct and divide the voter's absolute distance to their polling site by that median. 

\textbf{Load Normalization.} In order to normalize the site loads to easily digestible numbers, we divided each group's median site load by that of the majority group, i.e. White voters in both states, as baseline.

In Figure~\ref{fig:normalized_access} the normalized distances to voting locations and loads are presented for Florida and North Carolina voters. We should note that although called distances, the values in these tables are not actual distance but ratios. Therefore, they cannot be compared to the values in Figure~\ref{fig:raw_analysis} directly. Instead, our goal here is to bring into attention how considering normalized distances changes the ordering among different races in terms of their distance to polling locations. In the state of North Carolina for example, on average American Indian voters experience the largest distance to polling locations while Black voters have the shortest travel distance when absolute distances are considered (Figure~\ref{fig:raw_analysis}). However, when normalized distances are considered as presented in Figure~\ref{fig:normalized_access}, the conclusions change. When normalizing to the nearest school or library, Black voters face the longest average travel distance to their polling site and white voters the shortest in both Florida and North Carolina. In Florida, this pattern is the same when considering the median of school/library normalized distance. In North Carolina, this pattern changes when considering the median. Normalizing to voter density tells a slightly different story, particularly in North Carolina, where Asian voters experience the longest mean density-normalized travel distance and American Indian and Alaskan Native the shortest. This highlights the importance of careful consideration of how distance to polling location is measured. Different choices among reasonable alternatives can lead to different conclusions about which group is most privileged with respect to polling access. Ultimately, the numbers themselves cannot tell us which normalization scheme (if any) is most appropriate and we must rely on common sense and contextual understanding to make a subjective determination of which metric is most appropriate in this case. 

Though imperfect, we think the median of nearest school or library-normalized distance is the best metric for assessing access disparities based on distance to polling location. As discussed above, we believe this is a reasonable proxy for the distance a person regularly travels for everyday tasks. We prefer it to the median pairwise distance normalization because this method for normalization may behave unintuitively for residents of multi-family housing. For example, residents of a large apartment building would likely have very small median distance to other voters in their precinct, which may not reflect the distance they regularly travel. Going forward, we present our results in terms of the school/library-normalized distance.

\begin{figure}[htbp]
    \centering
    \begin{tabular}{cc}
    \textbf{Florida, normalized access} &
    \textbf{North Carolina, normalized access}\\

\centering
\begin{tabular}{@{}lccccc@{}}
\toprule
\multicolumn{1}{c}{\multirow{3}{*}{Race}} &
\multicolumn{4}{c}{Normalized Distance} &
Norm.\\
&
\multicolumn{2}{c}{school/lib.} &
\multicolumn{2}{c}{voter dist.} & Load \\
\cmidrule(l){2-3}\cmidrule(l){4-5} 
\multicolumn{1}{c}{} & Mean & Med. & Mean & Med. &\\
\cmidrule(l){1-6}
Wht. not Hisp. & 1.65 & 1.09 & 1.19 & 1.02 & 1\\
Hispanic & 1.86 & 1.26 & 1.35 & 1.10 & 0.89\\
Blk. not Hisp.  & 1.98 & 1.24 & 1.20 & 1.02 & 0.84\\
Unknown & 1.84 & 1.21 & 1.29 & 1.06 & 0.91\\
Asn./Pac. Isl. & 1.90 & 1.22 & 1.26 & 1.07 & 1.02\\
Other & 1.84 & 1.18 & 1.26 & 1.05 & 0.99\\
Multi-Racial & 2.02 & 1.22 & 1.26 & 1.05 & 0.98\\
Amer. Ind. & 1.76 & 1.15 & 1.15 & 1.01 & 0.95\\ \bottomrule
\end{tabular}

&

\centering
\begin{tabular}{@{}lccccc@{}}
\toprule
\multicolumn{1}{c}{\multirow{3}{*}{Race}} &
\multicolumn{4}{c}{Normalized Distance} &
Norm.\\
&
\multicolumn{2}{c}{school/lib.} &
\multicolumn{2}{c}{voter dist.} & Load \\
\cmidrule(l){2-3}\cmidrule(l){4-5}
\multicolumn{1}{c}{} & Mean & Med. & Mean & Med. &\\
\cmidrule(l){1-6}
White & 1.69 & 1.01 & 1.08 & 0.98 & 1\\
Black/AfAm. & 1.95 & 1.05 & 1.09 & 0.98 & 1.00\\
Undesignated & 1.83 & 1.04 & 1.11 & 0.99 & 1.06\\
Other & 1.83 & 1.06 & 1.14 & 1.02 & 1.10\\
Asian & 1.82 & 1.07 & 1.21 & 1.06 & 1.20\\
AmerInd./AKNat. & 1.91 & 1.00 & 0.92 & 0.90 & 0.77\\
Two+ races & 1.90 & 1.05 & 1.13 & 1.01 & 1.07\\
NatHaw./PacIsl. & 1.77 & 1.05 & 1.10 & 0.98 & 1.07\\ \bottomrule
\end{tabular}
\end{tabular}
\caption{Mean and median distance to polling location, normalized by distance to closest school/library or median distance to other voters in the precinct. Load is normalized to the median load experienced by the majority group.}
\label{fig:normalized_access}
\end{figure}


%% file: solution.tex
\section{Alternative selection for polling places: algorithms}
\label{sec:sol}

Having established a way to measure disparities in access to voting, can we determine where to place polling locations so as to reduce these disparities?  In this section, we propose algorithms to reduce access disparities in terms of both distance and load.  These methods are scalable and work on the large input sizes necessary to handle states such as Florida and North Carolina. We formulate this problem as different versions of the well-known $k$-median and $k$-center discrete clustering problems.  In both these problems, we are given a set of points $X$ (the voters) and facilities $F$ (polling locations) in a metric space and the goal is to choose $k$ facilities from $F$ and assign each point in $X$ to its nearest facility so that the overall cost (measured either as the sum of distances or the maximum distance) is minimized. In general, one assumes that $F = X$, but we will distinguish the two sets in our setting. 

We solve variants of the standard $k$-median and $k$-center algorithms with additional constraints that enforce the fair access criteria. Our approach here is an extension to the fair clustering literature which has gained momentum in the past few years \cite{makarychev2021approximation,chierichetti2017fair,kleindessner2019fair,kleindessner2019guarantees,chen2019proportionally,schmidt2018fair,ahmadian2019clustering,rosner2018privacy,vakilian2021improved,abraham2019fairness,negahbani2021better,Chlamatac2022Approximating}. To ensure the developed algorithms can be applied to large datasets, we employ the concept of \emph{coresets}. Given a problem (e.g. $k$-median), coresets are small, weighted subsets of large datasets such that the solutions to the problem found on subset are provably close to the solutions found on the original dataset \cite{agarwal2005geometric}. Our methods first summarize the massive voter dataset $X$ using coreset construction algorithms, and then feed the summarized versions into fair clustering algorithms to produce the desired result.

We introduce three solutions that provide more equitable alternatives to the original polling location assignment, with respect to the proposed access disparity measures. In the first solution, we  minimizing the maximum distance to polling places across different race categories while ignoring load balance. In the second solution we build upon the first method to find an assignment which also provides a more balanced distribution of voters, at the cost of needing to open up more polling locations. In our final solution, we discuss an alternative approach to addressing both distance and load issues simultaneously without violating the adding new facilities. All these solutions come with their own advantages and limits.

\subsection{Group fair distances}
\label{sec:distance_fair}

The problem of selecting a predetermined number of polling places and assigning exactly one to each voter, so that the overall distance between voters and polling places is minimized, can be formulated via the well-known problem of $k$-median clustering.
The objective in the $k$-median algorithm is to select $k$-centers (i.e polling locations) so as to minimize the sum of the distances between points and their associated center:
\[\argmin_{C}\sum_{p \in X}|| p - C(p) ||\]
where $C(p)$ denotes the center point $p$ is assigned to.

This formulation of the $k$-median problem can result in arbitrarily large distances for certain sub-populations within voters, as long as the overall average distance is minimized. This is a problem in the context of polling as it may  hinder certain individuals' ability to cast their vote, due to longer than necessary travel distances. We formulate this via a variant of the \emph{fair $k$-median} clustering problem introduced in \cite{abbasi2021fair} and introduce an algorithm to address this issue.

\begin{defn}[Fair $k$-median clustering]
\label{def:fairkm}
Given $m$ groups $X = X_1 \cup \dots \cup X_m$, a fair $k$-median clustering algorithm returns $k$ centers so as to minimize the maximum average distance to centers across all groups.
\[\argmin_{C \in \mathcal{C}}~ \text{max} \left( \frac{1}{|X_1|}\cost_C(X_1), \dots, \frac{1}{|X_m|}\cost_C(X_m) \right) \]
where $\mathcal{C}$ is the set of all possible choices of cluster centers, and $\cost_C(X_i)$ denotes the sum of distances for group $X_i$ in clustering $C$.
\end{defn}

In the context of polling and the concern around racial disparities, each group $X_i$ in the above definition represents voters of a particular race, which is determined according to the race categories in the voter rolls.
To solve the fair $k$-median clustering problem Abbasi \etal first solve an associated linear program 
and use a faithful rounding procedure to choose exactly $k$ centers. When voters' addresses can be selected as centers, this method results in expected approximation factor of 4 \cite{abbasi2021fair}.
Unfortunately, directly using this linear programming-based method is not realistic because it would be prohibitively slow or even impossible when determining polling locations for millions of voters.  Therefore, instead of processing the entire input data we propose to use its \emph{coreset}\cite{agarwal2005geometric}. Denote the average cost of standard and fair $k$-median clustering on point set $X$ as $\cost_C(X)$ and $\cost^f_C(X)$ respectively, where $C$ is a set of $k$ centers. 
\begin{defn} [$(k, \epsilon)$-coreset for $k$-median] Given a point set $X$ in a metric space, the $(k, \epsilon)$-coreset for $k$-median is a weighted subset $S$ of $X$, where for each $C$ with size $k$:
\[(1-\epsilon)\cost_C(X) \le \cost_C(S) \le (1+\epsilon)\cost_C(X)\]
\end{defn}

\begin{theorem}
\label{th:coreset}
Given point set $X = X_1 \cup ... \cup X_m$ as input, let $S_i$ be a $(k, \epsilon)$-coreset for group $X_i, i \in 1,...,m$, separately. Then for all $C$:
\[(1-\epsilon)\cost^f_C(X) \le \cost^f_C(S) \le (1+\epsilon)\cost^f_C(X)\]
where $S = S_1 \cup ... \cup S_m$.
\end{theorem}
\begin{proof}
We know by definition that $\cost^f_C(S) = \argmax_{i \in 1...m} \cost_C(S_i)$. 
Since $S_i$ is a $(k, \epsilon)$-coreset we have that  for $i \in 1...m$, $(1-\epsilon)\cost_C(X_i) \le \cost_C(S_i) \le (1+\epsilon)\cost_C(X_i), \forall C$. Therefore
\begin{align*}
(1 - \epsilon) \argmax_{i \in 1...m} \cost_C(X_i) & \le \argmax_{i \in 1...m} \cost_C(S_i) \le (1 + \epsilon) \argmax_{i \in 1...m} \cost_C(X_i)
\end{align*}

or:
\[(1-\epsilon)\cost^f_C(X) \le \cost^f_C(S) \le (1+\epsilon)\cost^f_C(X)\]
\end{proof}

Theorem~\ref{th:coreset} suggests that given a large point set as input one can run any fair $k$-median algorithm (as defined in Definition~\ref{def:fairkm}) on the union of its groups' coresets, and find a clustering with an objective value arbitrarily close (within $\epsilon$ factor) to that of running the same algorithm on the original set. 

For our coreset construction we use the algorithm proposed by Feldman and Langberg in \cite{feldman2011unified}.  In the first step, it  uses a bi-criteria $k$-median clustering algorithm as a subroutine to find initial centers, and assign weights to each point in the original set based on its distance to the closest center\footnote{An $(\alpha, \beta)$ bi-criteria $k$-median clustering algorithm opens up to $\beta \times k$ centers, which results in an objective cost smaller than $\alpha$ times the optimal solution.}. In our implementation for $(\alpha, \beta)$ bi-criteria $k$-median algorithm we use the method due to Indyk\cite{indyk1999sublinear}. In the second step the points are sampled according to the distribution implied by the assigned weights, and their union with bi-criteria centers are returned as a weighted coreset. If the original point set is defined in a metric space  (which is the case in our dataset), then with probability $1-\delta$ this algorithm returns a weighted $\epsilon$-coreset of size $\frac{c}{\epsilon^2}(k \log(n) + \log(1/\delta))$ where $n$ is the input size and $c$ is a large enough constant. In our implementation this algorithm is used to construct a coreset for every group and the union of these coresets is fed into the fair clustering algorithm to find a distance fair assignment. An illustration of this method is provided in Figure~\ref{fig:feldman}.

\begin{figure*}[htb]
\centering
\includegraphics[width=0.7\textwidth]{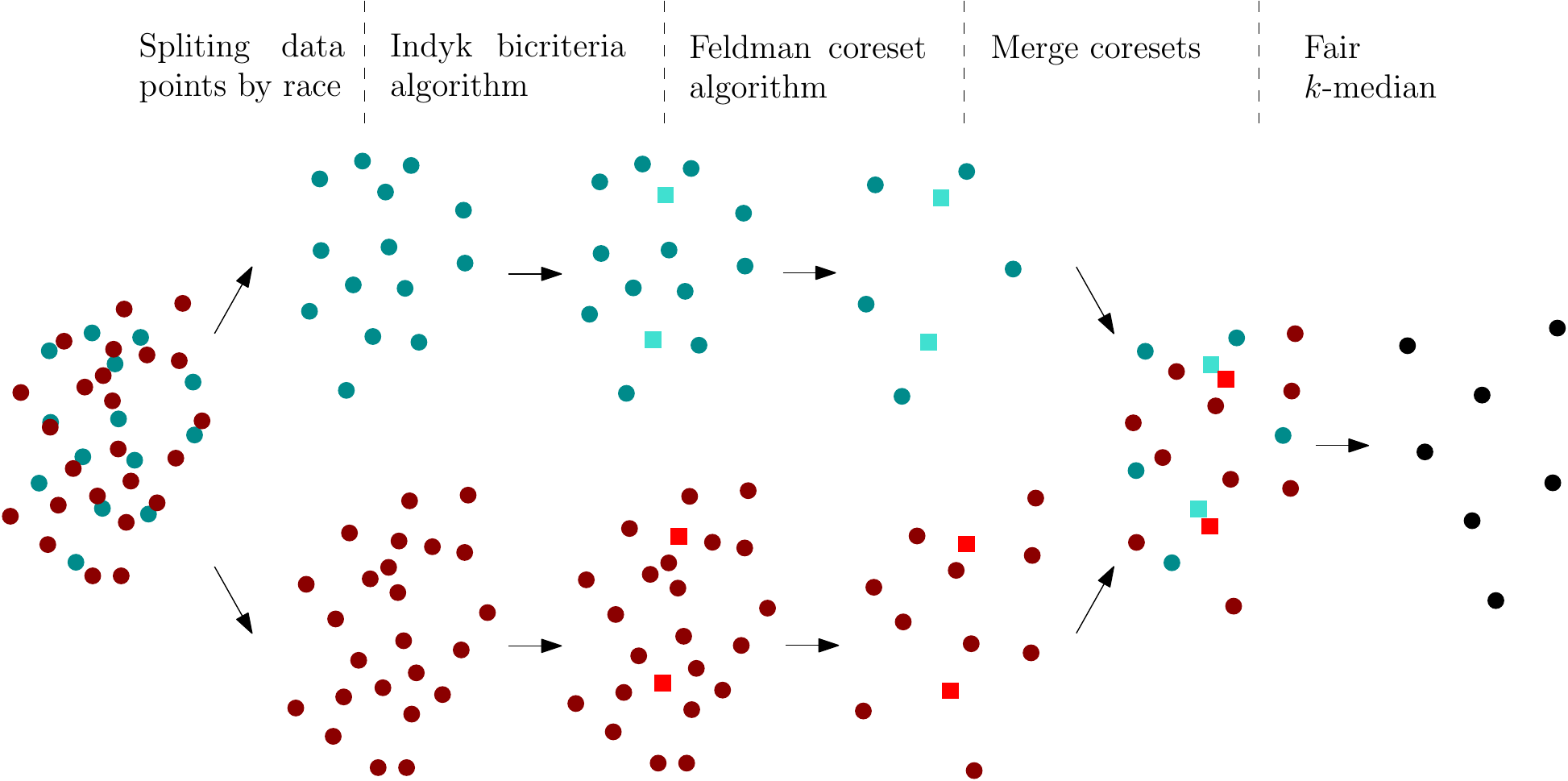}
\caption{Fair $k$-median clustering algorithm using coresets for two groups. Data points are first split based on their group. The coreset for each group is constructed as described and then the union of coresets is used as input for the fair $k$-median algorithm.}    
\label{fig:feldman}
\end{figure*}

In order to empirically evaluate the effectiveness of using coresets from Theorem~~\ref{th:coreset}, we compare the results of running regular and fair $k$-median clustering algorithms on a sample dataset, to the corresponding values achieved on its coreset. For this purpose we randomly selected 4000 and 1000 White and Black voters from North Carolina voter records, respectively. The results of regular and fair $k$-median algorithms on this sample and its coreset are summarized in Appendix Figure~\ref{fig:coreset}; the objective values achieved by running both algorithms on the coreset are very close to the values from the entire sample data, which corroborates the results of Theorem~\ref{th:coreset}. Both theoretical and empirical results demonstrate that we can use coresets in our analysis and achieve near optimal results.  We return to assessing this algorithm on our Florida and North Carolina data in Section \ref{sec:experiments}.

\subsection{Group fair distances and balanced assignments}
\label{sec:load_fair}

A fair $k$-median framing of the problem of access can address concerns around traveling to a polling station. But as we have pointed out earlier, another component of the overall time it takes to vote is the load at the polling station itself. 



The problem of opening and assigning polling locations to voters so that the overall distance to polling places is minimized while maintaining a balanced distribution of voters per location, can be formulated via the \emph{balanced} $k$-median clustering problem\footnote{This problem is also known as capacitated $k$-median clustering problem} where each facility comes with a capacity constraint on the number of points that can be assigned to it. Previous algorithms for the balanced $k$-median either violate the capacity constraint or the cardinality (number of centers) constraint. In this section we opt for a method which may open more centers but maintains the capacity constraint.
We start by assuming that all facilities have a uniform capacity $L$
We call the fair variant of this problem \emph{$L$-balanced fair $k$-median clustering}. Bateni \etal \cite{bateni2014distributed} demonstrated how an approximate solution to the $k$-median clustering problem can be transformed into an approximate bicriteria solution for the $L$-balanced $k$-median clustering with slightly worse approximation factors. We closely follow their approach and show this also holds true for the fair $k$-median variant of the problem in the following theorem. 
\begin{theorem}
\label{th:balanced}
Suppose there is an $\alpha$ approximation algorithm for the unconstrained (no balance constraint) fair $k$-median problem. Then there exists a $(2\alpha, 2)$ bicriteria approximation for $L$-balanced fair $k$-median problem.
\end{theorem}
\begin{proof}
The proof of this theorem is similar to that of Bateni \etal, with a slight variation to accommodate for the fair version of the $k$-median algorithm. Given an instance of $L$-balanced fair $k$-clustering, we first solve the unconstrained problem and obtain an $\alpha$ approximate solution. If each cluster in this solution has at most $L$ points we are done. Otherwise, consider a cluster with center $u$ that has $N_u$ points where $N_u > L$. This cluster is broken down into $\ceil*{\frac{N_u}{L}}$ clusters and the center for each new cluster will be its closest point to $u$. The points in the original cluster will be assigned to one of the new $\ceil*{\frac{N_u}{L}}$ centers as to satisfy the capacity constraint. Now consider one of the new clusters created in this way and assume its center is $u'$. Since for every point $v$ in the new cluster $d(u',u) \le d(v,u)$ then $d(v, u') \le 2d(v,u)$. On the other hand, we know the original solution is $\alpha$ approximate, i.e. the cost for the group with the highest average distance is $\alpha$ times the corresponding value in the optimal solution. Since the distance between all points and their respective centers in the updated solution has been at most doubled compared to the original solution, the updated solution is a $2\alpha$ approximate solution. 
The updated solution has $\sum_u\ceil*{\frac{N_u}{L}} < \sum_u \frac{N_u}{L} + 1$ clusters. Since the number of clusters in the original solution cannot be more than $k$, then at most $k$ new clusters have been opened which means the number of clusters in the updated solution is less than or equal to $2 k$.
\end{proof}
Theorem~\ref{th:balanced} gives us an $8 + \epsilon$ expected upper-bound to an instance of $L$-balanced fair $k$-median problem where at most $2k$ centers are opened.  We assess these results in practice on the Florida and North Carolina data in Section \ref{sec:exp_balfair_med}.

\subsection{Individually fair polling assignment}
\label{sec:kcenter}
The algorithms discussed in sections~\ref{sec:distance_fair} and \ref{sec:load_fair} both consider a group notion of fairness to deliver a polling place assignment with more equitable assignments by group. But we can also investigate the issue of inequitable distances to polling locations at an \emph{individual} level. That is, the objective would be to find a solution where no single voter, irrespective of others, is too far away from the nearest polling place. In this section we discuss a method to minimize the maximum distance a voter has to travel to cast their vote.

To address the issue described above we look at another well-known variant of the $k$-clustering problem called $k$-center. Given a set of data points and parameter $k$, the objective in the $k$-center problem is to select $k$ centers from the input and assign each data point to one of the selected centers, so the maximum distance between the points and their assigned centers is minimized, i.e. $\argmin_{C \in \mathcal{C}} \argmax_{p \in X} \|p-C(p)\|$
where $X$ is the set of data points and $\mathcal{C}$ is the set of all possible $k$-clusterings.

In our analysis we consider two variants of the $k$-center problem, without center load constraints and with. As before we assume a uniform load requirement for each center. Our solution is similar: we build a core-set and then use an unconstrained algorithm to cluster the coreset. To build the coreset we use the distributed weighted balanced $k$-center algorithm proposed in \cite{mirjalali2020improved} that splits the data into chunks and runs the well-known $2$-approximation due to Gonzalez\cite{gonzalez1985clustering} to generate a coreset. Once we have collected and weighted the coresets appropriately, we run a capacitated $k$-center algorithm that is a weighted variant of the algorithm proposed in \cite{khuller2000capacitated}. In their paper, Mirjalali \etal show that this process yields a $49$-approximation overall.

%% file: experiments.tex
\section{Alternative selection for polling places: experiments}
\label{sec:experiments}

We next assess the effectiveness of the fair polling site selection methods introduced in Section \ref{sec:sol} to analyze the impact they could have on reducing voting disparities in Florida and North Carolina.
We evaluated these algorithms based on two options for new polling locations.
In the first setting, polling places can be opened at any voter residence. This setting is used to evaluate a best case scenario and is less realistic.\footnote{Although unusual, there are some jurisdictions that do allow residences (garages) to be designated polling locations, precisely due to restrictions on the number of people per precinct (which we term ``load"): \url{https://www.good.is/articles/polling-place-garages-san-francisco}}
In the second setting, the algorithm can only select the polling places from a set containing the state's schools, libraries, and 2020 election designated polling places.  This gives a realistic bound on improvement from better polling location placement, since schools and libraries are often used as polling locations.
The full set of numerical results for all algorithms, including both polling selection methods, normalized and real valued results for each racial group, are given in the Appendix. Here, we focus on the normalized median distance based on the polling selection method that allows selection from school, library, and existing polling site locations, and on the normalized site loads.  Recall from Section \ref{sec:polling-stats} that we see this type of normalized distance as the most useful measure of voter distance access assessed.  Interestingly, even though the algorithmic guarantees and optimizations introduced in Section \ref{sec:sol} are based on the un-normalized distances and polling site loads, we find that the fair algorithms perform well on the normalized measures.

\begin{figure}[htbp]
    \centering
    \begin{tabular}{cc}
    \includegraphics[width=2.9in]{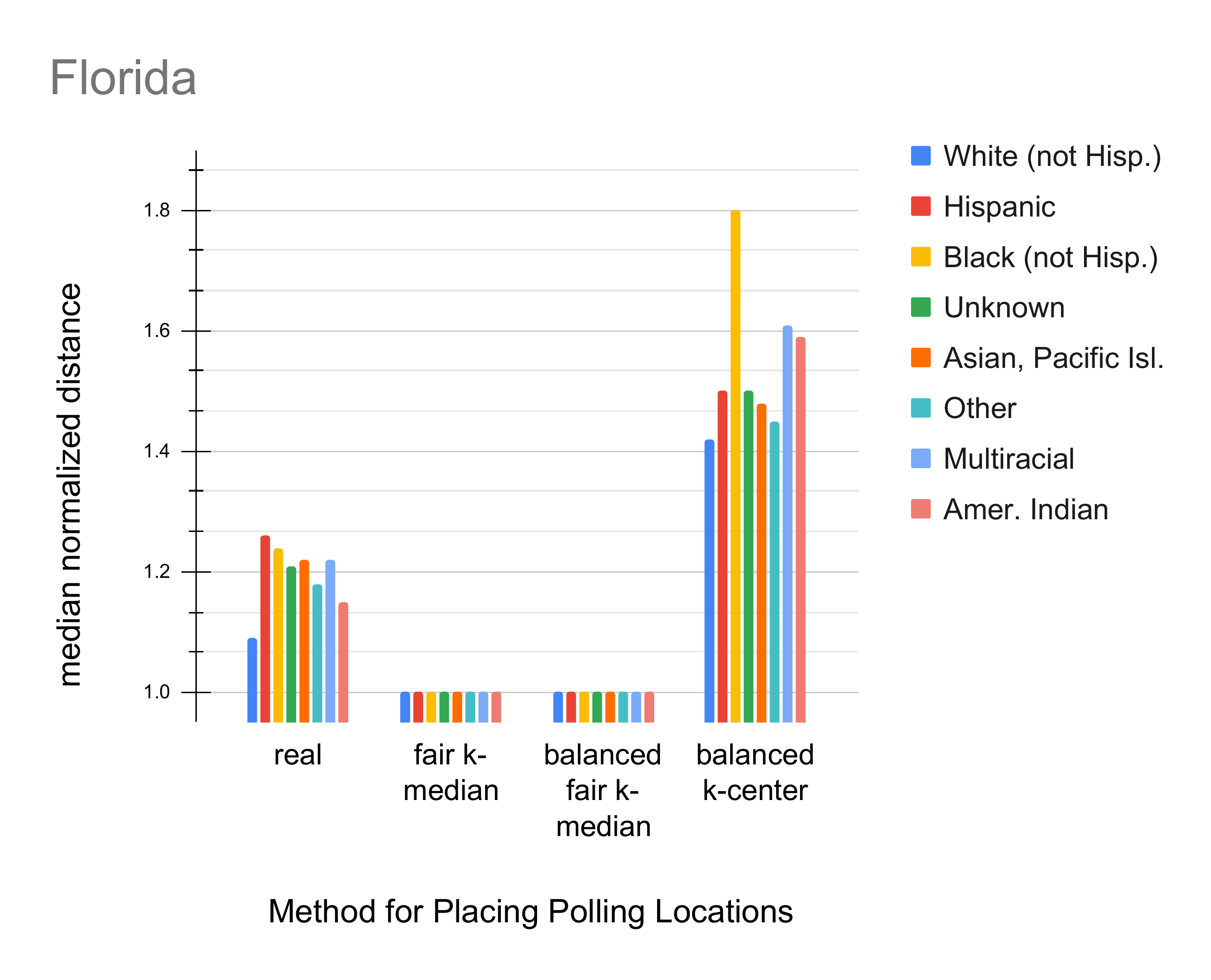} &
    \includegraphics[width=2.9in]{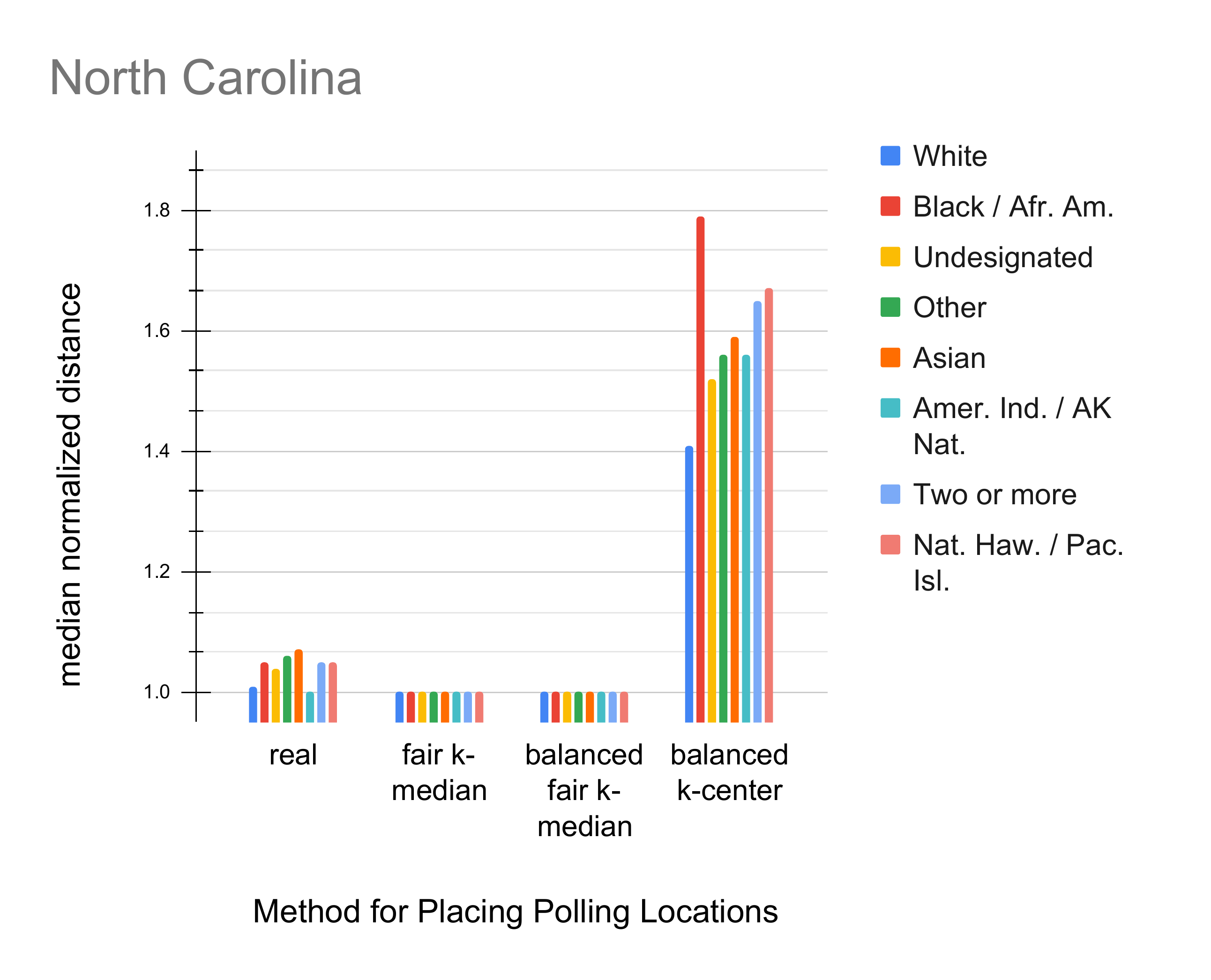}\\
    \includegraphics[width=2.9in]{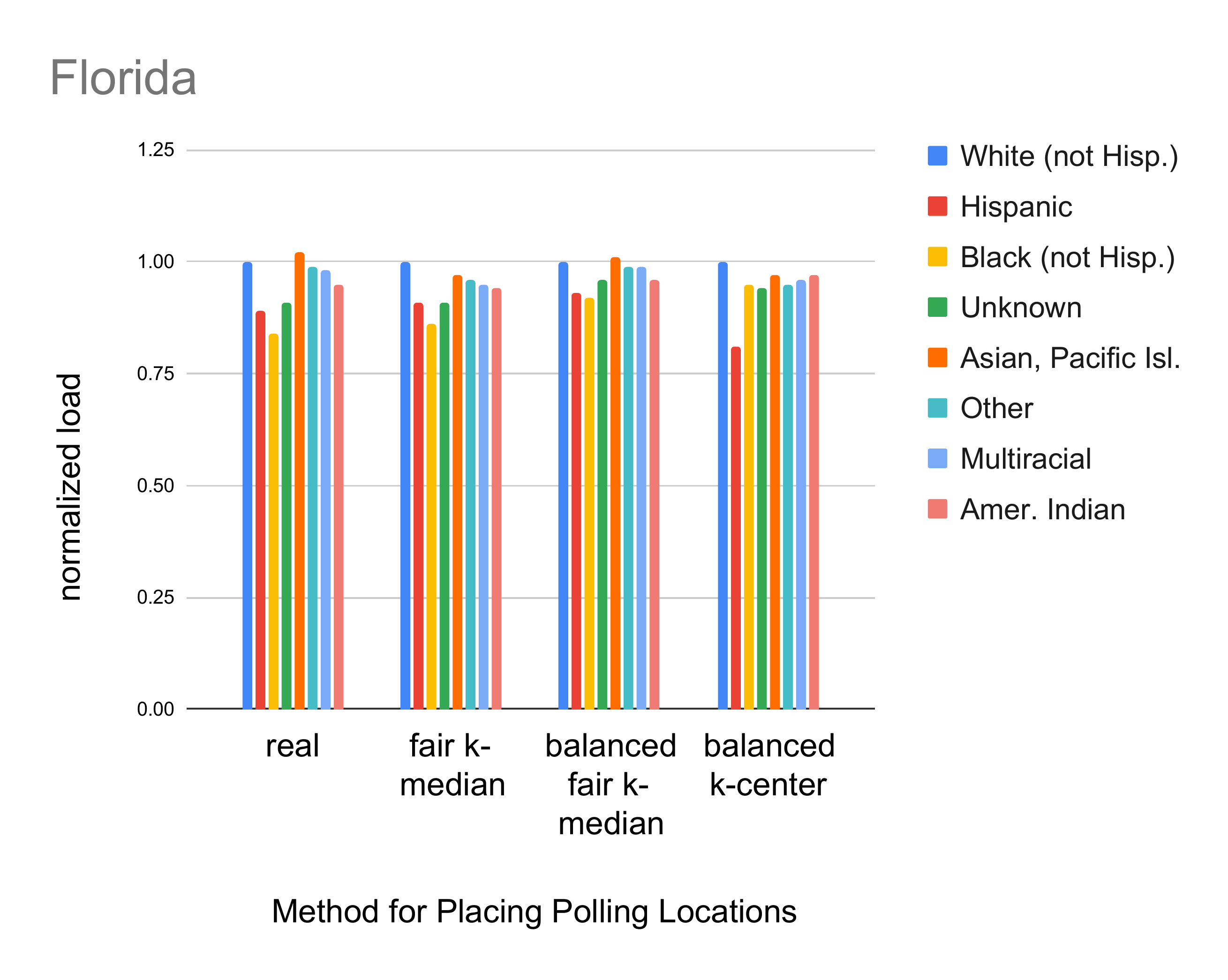} &
    \includegraphics[width=2.9in]{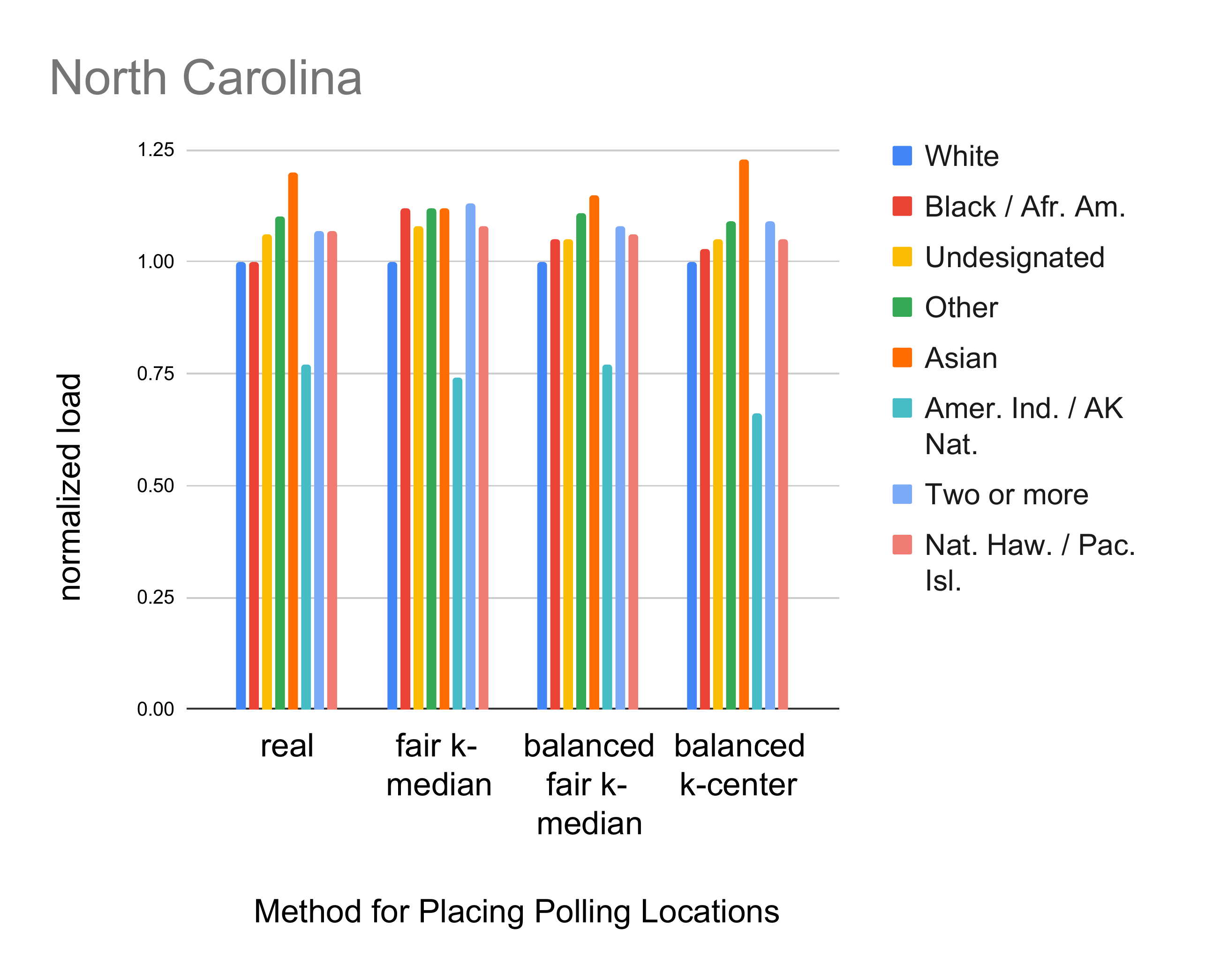}
    \end{tabular}
    \caption{Comparative results for the real polling locations (based on analysis in Section \ref{sec:polling-stats}) and fair polling location selection algorithms introduced in Section \ref{sec:sol}.  Results are given for Florida (left) and North Carolina (right) based on normalized median values for distances (top) and normalized polling load values (bottom). The polling selection algorithm results shown here allow polling sites to be opened at schools or libraries. Balanced fair $k$-median and balanced $k$-center results are for capacity parameter $\epsilon=0.5$.}
    \label{fig:alg_normalized_comparison}
\end{figure}

The results, shown in Figure \ref{fig:alg_normalized_comparison}, show that both of the fair $k$-median variants are able to achieve a median normalized distance of $1$ for all racial groups across both states.  This means that the median voter travels only as far as their closest school or library under these fair algorithm variants.  In North Carolina, the real polling location assignments were very close to also matching this result.  However, in Florida, the new polling locations from the fair $k$-median variants would decrease the distance-based access disparities between groups while also decreasing these distances for all groups.  Hispanic and Black voters, who experienced the largest normalized distance, had to travel about 25\% further than the closest school or library under the real polling location assignments and can have that reduced under the fair assignments.  These reductions can translate into substantial distances for voters.  For example, the difference between the minimum and maximum group median distances within Glades County in Florida is reduced from 6.26 miles in the original assignment to 1.21 miles in the assignment produced by the fair $k$-median algorithm.

The normalized load results demonstrate that the balanced fair $k$-median algorithm is the most effective at balancing the load across groups, although the other methods do somewhat help to achieve balance.
%
We considered three different capacities for both the balanced fair $k$-median and the balanced $k$-center algorithms. These capacities were $1 + \epsilon$ times the average number of voters per location, where $\epsilon \in \{0.1, 0.5, 0.9\}$ is the control variable, and average is determined by dividing the the total number of voters by the total number of designated polling places. The results in Figure \ref{fig:alg_normalized_comparison} are for $\epsilon=0.5$.  
Recall that this capacity choice sets the allowed deviation from a fixed (balanced) load across polling sites.  Thus, the capacity choice impacts the number of extra polling places opened up by the balanced fair $k$-median algorithm, and voter distribution balance in the assignments produced by both balanced fair $k$-median and balanced $k$-center algorithms. The results are summarized in Tables~\ref{tab:both_balanced_epsilon} and \ref{tab:both_kcenter_epsilon}. As expected, we see by increasing the capacity fewer additional polling sites need to be opened for the balanced fair $k$-median algorithm. For both algorithms we see that increasing the capacity results in a larger standard deviation of the number of voters per location, which means voters are less evenly distributed across polling sites.

\begin{table}[htb]
\caption{The effect of capacity in $L$-balanced fair $k$-median algorithm on the number of newly opened polling sites and voter distribution}
\label{tab:both_balanced_epsilon}
\centering
\begin{tabular}{@{}lcccccc@{}}
\toprule
\multicolumn{1}{c}{\multirow{3}{*}{ }} & \multicolumn{3}{c}{\begin{tabular}[c]{@{}c@{}}Florida\end{tabular}} & \multicolumn{3}{c}{\begin{tabular}[c]{@{}c@{}}North Carolina\end{tabular}} \\ \cmidrule(l){2-4} \cmidrule(l){5-7}
\multicolumn{1}{c}{}    & $\epsilon=0.1$      & $\epsilon=0.5$      & $\epsilon=0.9$        & $\epsilon=0.1$   & $\epsilon=0.5$    & $\epsilon=0.9$   \\ \midrule
Number of extra polling sites & 1472 & 690 & 276 & 915 & 432 & 207\\
Mean \#voters per location & 2404 & 2854 & 3173 & 1871 & 2252 & 2487\\
Std. dev. of number of voters per location & 1118 & 1378 & 1596 & 860 & 1106 & 1317\\ \bottomrule
\end{tabular}
\end{table}

\begin{table}[htb]
\caption{The effect of capacity in balanced $k$-center algorithm on voter distribution}
\label{tab:both_kcenter_epsilon}
\centering
\begin{tabular}{@{}lcccccc@{}}
\toprule
\multicolumn{1}{c}{\multirow{3}{*}{ }} & \multicolumn{3}{c}{\begin{tabular}[c]{@{}c@{}}Florida \end{tabular}} & \multicolumn{3}{c}{\begin{tabular}[c]{@{}c@{}}North Carolina\end{tabular}} \\ \cmidrule(l){2-4} \cmidrule(l){5-7}
\multicolumn{1}{c}{}    & $\epsilon=0.1$      & $\epsilon=0.5$      & $\epsilon=0.9$        & $\epsilon=0.1$   & $\epsilon=0.5$    & $\epsilon=0.9$   \\ \midrule
Mean \#voters per location & 4628 & 4628 & 4628 & 3561 & 3561 & 3561\\
Std. dev. of number of voters per location & 1503 & 2331 & 2806 & 1255 & 1821 & 2150\\ \bottomrule
\end{tabular}
\end{table}

These methods have distinct strengths and weaknesses and are suitable for different use cases. While the fair $k$-median algorithm produces an assignment with more equitable distances at a group level as well as offering shorter distances overall, it does not take into consideration the potentially unbalanced loads at polling sites.  Still, in practice it performs well when assessing normalized load. The $L$-balanced fair $k$-median algorithm solves this issue by introducing a limit on the number of voters that can be assigned to a single facility. Although this method may need to open up additional polling places, this requirement could instead be used to allocate additional resources, e.g. placing more polling booths at certain polling sites. The unconstrained (Appendix~\ref{app:kcenter}) and balanced $k$-center algorithms introduced in Section \ref{sec:kcenter} address equitable access at an individual level. Similar to the $k$-median methods, the main difference between the two is that one distributes voters evenly across polling paces while the other does not provide such guarantee. The nice property of the balanced $k$-center algorithm is that it can be tuned to specify exactly how many voters above the polling place limit could be assigned to it. However, this comes at the expense of less competitive distances to polling sites.   

%% file: conclusion.tex
\section{Limitations and Conclusion}
\label{sec:conclusions}
In this paper, we analyzed voting access disparities with respect to polling locations.  We introduced a methodology to study these voting access disparities, focusing on quantifying potential racial disparities in terms of distance to nearest polling location and number of people assigned to a given polling location (its ``load"). To account for natural variations in population density that might give, e.g., individuals in rural locations the expectation and means of traveling further to the polls, we normalized the distance to nearest polling location by the distance to the nearest school or library as well as by another natural measure of population density (the median distance to another voter).  Examining these disparities for all voters in Florida and North Carolina in the 2020 US general election, we found that non-white voters had to travel farther to the polls (using the normalized distances) in Florida while in North Carolina the distances to the polls showed fewer racial disparities. The number of voters assigned to a polling location also varied by race, though it was less clear that difference was substantive. However, the specific measurement choices for assessing voter access are critical; different reasonable measures of distance to polling location return different voting access disparity results.

These voting access disparity results are subject to a number of limitations and should be seen as only the beginning of an investigation into voting access disparities.  The measurement methodology we introduce here is able to quantify (normalized) distance to the nearest polling location and the polling location load, but a voter's experience of access to their polling location depends on many additional factors.  These include factors directly related to the polling location (e.g., number of voting stations or machines, time to wait in line, accessibility of the location by public transit and for voters with disabilities) as well other societal limitations that may keep people from the polls (e.g., availability of childcare, time off of work to vote, voter intimidation).  While distance and load may be reasonable proxies for some of these measures, they do not capture the full set of barriers that may prevent someone from voting.  Thus, these measures are most useful as a beginning point by identifying racial disparities that should be addressed, and lack of identification of such disparities in these two measures should not be considered a sign that voting access has been equalized. 

Additionally, we introduced multiple algorithmic interventions to assign polling locations to reduce racial disparities in distance and load.  These new algorithms could allow elections officials to place polling locations more effectively based on a given list of public location (we use schools and libraries).  They also allow the study of potential gains and trade-offs associated with possible alternations such as reassigning voters and moving or adding polling sites, but they assume each voter can be assigned to any of the selected polling sites and do not take existing precincts into account.  Additionally, these algorithmic interventions focus only on polling locations and therefore ignore other interventions that may be helpful in reducing voter access disparities, such as arranging rides to the polls or giving workers time off to vote. Universal vote-by-mail could be even more effective at alleviating voting access disparities, as it entirely avoids the distance and load concerns we raise here.
Within the bounds of existing in-person voting, however, the introduced algorithms can serve as a useful first step in mitigating racial disparities in voting access.

